\documentclass[12pt,a4paper]{article}
\usepackage{amsmath}
\usepackage{graphics}
\usepackage{graphicx}
\usepackage{epsfig}
\usepackage[utf8]{inputenc}
\usepackage{floatrow}
\usepackage{float}
\usepackage{amssymb}
\usepackage{amsthm}
\usepackage{graphicx}
\usepackage{epstopdf}
\usepackage{anysize}
\usepackage{geometry}
\usepackage{setspace}
\usepackage{longtable}
\usepackage{verbatim}

\usepackage{subfig}
\usepackage{xcolor}
\usepackage{color,soul}
\usepackage{hyperref}
\usepackage{cite}
\hypersetup{colorlinks,linkcolor={blue},citecolor={blue},urlcolor={black}}

\marginsize{1cm}{1cm}{1cm}{1cm}
\input epsf.sty
\usepackage{lineno}
\usepackage[english]{babel}
\begin{document}

\newtheorem{theorem}{Theorem}[section]
\newtheorem{lemma}[theorem]{Lemma}
\newtheorem{example}[theorem]{Example}
\newtheorem{proposition}[theorem]{Proposition}
\newtheorem{corollary}[theorem]{Corollary}
\newtheorem{remark}[theorem]{Remark}
\newtheorem{definition}[theorem]{Definition}
\newtheorem{result}[theorem]{Result}
\newtheorem{numerical example}[theorem]{Numerical example}

\begin{center}
	{\bf Mutations make pandemics worse or better: modeling SARS-CoV-2 variants and imperfect vaccination}
\end{center}

\begin{center}
	\baselineskip .2in {\bf Sarita Bugalia$^1$, Jai Prakash Tripathi$^{1}$, Hao Wang$^2$}
	
	{\small $^1$ Department of Mathematics, Central University of Rajasthan,
		\\ Bandar Sindri, Kishangarh-305817, Ajmer, Rajasthan, India\\
	$^2$ Department of Mathematical and Statistical Sciences,\\ University of Alberta, Edmonton AB T6G 2G1, Canada}
	
\end{center}

\begin{abstract}
	COVID-19 is a respiratory disease triggered by an RNA virus inclined to mutations. Since December 2020, variants of COVID-19 (especially Delta and Omicron) continuously appeared with different characteristics that influenced death and transmissibility emerged around the world. To address the novel dynamics of the disease, we propose and analyze a dynamical model of two strains, namely native and mutant, transmission dynamics with mutation and imperfect vaccination. It is also assumed that the recuperated individuals from the native strain can be infected with mutant strain through the direct contact with individual or contaminated surfaces or aerosols. We compute the basic reproduction number for each strain independently and take the maximum for $R_0$. We prove the nonexistence of backward bifurcation using the center manifold theory, and global stability of disease-free equilibrium when the basic reproduction number $R_0<1$, that is, vaccine is effective to eliminate the native and mutant strains even if it cannot provide full protection. Hopf bifurcation appears when the endemic equilibrium loses its stability. An intermediate mutation rate $\nu_1$ leads to oscillations. When $\nu_1$ increases over a threshold, the system regains its stability and exhibits an interesting dynamics called endemic bubble. An analytical expression for vaccine-induced herd immunity is derived. The epidemiological implication of the herd immunity threshold is that the disease can be effectively eradicated if the minimum herd immunity threshold is attained in the community. Furthermore, the model is parameterized using the Indian data of the cumulative number of confirmed cases and deaths of COVID-19 from March 1 to September 27 in 2021, using MCMC method. The cumulative cases and deaths can be reduced by increasing the vaccine efficacies to both native and mutant strains. We observe that by considering the vaccine efficacy to native strain as 90\%, the cumulative cases and deaths would be reduced by 3.27\% and 5.2\%, respectively; and by considering the vaccine efficacy to mutant strain as 90\%, the cumulative cases and deaths would be reduced by 0.9\% and 2.5\%, respectively. Our study demonstrates that the COVID-19 pandemic may be worse due to the occurrence of oscillations for certain mutation rates (i.e., outbreaks will occur repeatedly) but better due to stability at a lower infection level with a larger mutation rate. We perform sensitivity analysis using the Latin Hypercube Sampling (LHS) methodology and partial rank correlation coefficients (PRCCs) to illustrate the impact of parameters on the basic reproduction number, the number of cumulative cases, and the number of deaths, which ultimately sheds light on disease mitigation. 
\end{abstract}

\section{Introduction}
Contagious diseases are one of the foremost reasons for demise worldwide. The spread of contagious diseases dangerously affects the growth of countries and the evolution of a population. Though modern scientific medicine has made rapid advancements, the diseases have not been completely eradicated. Diseases have obtained new versions due to the genetic variations of pathogens triggered via mutations. Many pathogens are characterized by more than one variant \cite{Sato1995,Palese1982}. Virus or pathogens mutations are general in contagious diseases such as HBV \cite{Sato1995}, Influenza \cite{ Palese1982}, and HIV \cite{Eron1998}. Multi-strain of 1918 avian influenza virus with the mutation have been recognized by Iwami et al. \cite{Iwami2007}. The existence of different variants of a pathogen is mainly due to resist immune attacks of the host or induced by treatment with antiviral drugs or antibodies \cite{Eron1998}. Ultimately, they confirm the persistence of disease in a host. Sansonetti and Arondel \cite{Sansonetti1989} have revealed that mutant strains can be associated with higher virulence to disease than the native strains, and those people diseased with mutant strains have a higher death rate in the contagious diseases such as plague, influenza A, etc. Thus, one of the major challenges in stopping the spread of infectious diseases is to treat with the genetic variations of pathogens \cite{MaRM1995, PartonR1994, Liu2018}. Mathematical models are helpful to describe and understand the dynamics of different strains under mutation. Various epidemic models with multi-strain contacts and mutation have been proposed in recent years from different aspects \cite{Liu2018, Cai2012, Li2004, Martcheva2007}. Liu et al. \cite{Liu2018} proposed a mathematical model for Influenza with virus mutation and analyzed the model in the sense of permanence of the disease. Cai et al. \cite{Cai2012} proposed a two-strain model with vaccination. The authors studied the existence and stability of the equilibria as well as the existence of Hopf bifurcation from endemic equilibria. Li et al. \cite{Li2004} proposed a two-strain SIR model with infection age and mutation. The authors analyzed the stability of equilibria and Hopf bifurcation. Moreover, epidemiological investigations have exposed that the phenomenon of mutations leads to further resistant viruses giving the emergence of many new dangerous epidemics or even new serious pandemics.

Since December 2019, a novel coronavirus (identified to produce a respiratory disease acknowledged as COVID-19) has spread broadly and rapidly and has since oppressed a significant proportion of the population worldwide. The World Health Organization (WHO) has detected the spread of COVID-19 as a pandemic, and as of December 4, 2021, over 263 million peoples were diseased, and about 5.2 million died of the virus. The SARS-CoV-2 virus triggered by severe acute respiratory syndrome is also mutating. Lately, numerous variants of the SARS-CoV-2 virus have been identified. These variants are described according to the number and types of mutations \cite{Korber2020, Lemieux2021}. Since the start of the COVID-19 pandemic, the SARS-CoV-2 coronavirus that initiated COVID-19 has mutated, rising in different variants of the virus. Numerous SARS-CoV-2 variants have developed worldwide, and the presence of different variants depends on several factors. One of these is called the delta variant, which was first identified in India \cite{CDC1}. Different variants have appeared in Brazil, England, California, and other countries. More transmittable variants such as beta, which first emerged in South Africa, may have improved the ability to re-infect individuals who have recuperated from previous versions of the virus and also be somewhat resistant to some of the coronavirus vaccines in development \cite{Hopkins}. These new variants might have distinct features that can influence the death rate and transmissibility \cite{Korber2020,Lemieux2021,Gonzalez2021}. From October 2020, the number of infected cases of SARS-CoV-2 and related deaths augmented drastically in England. It has been discovered that the new SARS-CoV-2 variant VOC-202012/01 was widespread, and its proportion amplified throughout the latest months in England \cite{Bussiness2021, PHE2021}. The mutations of viruses are frequent, and, as an outcome, SARS-CoV-2 can develop mutations with immunological resistance and fitness advantages \cite{Korber2020}. It is anticipated that further mutations will occur worldwide and probably even more after worldwide vaccination due to mutation force \cite{Rahimi2021}. 
Therefore, analyzing the effect of new strains of the SARS-CoV-2 virus is supremely significant. In the literature, a few authors proposed multi-variant mathematical models for COVID-19 \cite{Gonzalez2021, Khyar2020, Arruda2021}. Gonzalez-Parra et al. \cite{Gonzalez2021} studied the effect of a new, more infectious SARS-CoV-2 variant (VOC-202012/01 of lineage B.1.1.7) on hospitalizations, prevalence, and deaths associated with the SARS-CoV-2 virus.  Khyar et al. \cite{Khyar2020} proposed a multi-strain SEIR model with general incident rate and studied the global dynamics of the model. The authors also discussed the quarantine strategy for controlling the disease spread and fit the model to the Moroccan clinical data of COVID-19. Arruda et al. \cite{Arruda2021} proposed a model for COVID-19 and studied the optimal control of multi-strain epidemics.

Vaccination has been an effective strategy in battling the spread of contagious diseases, e.g., measles, influenza, and pertussis. In history, the elimination of smallpox has been counted as the most notable victory of vaccination ever recorded \cite{WHO}. 
Several authors in multiple papers have investigated the role of vaccination \cite{Cai2012, Martcheva2007, Castillo2002, Scherer2002, Liu2008, ArinoJ2003, Gumel2006, Cai2018, Alexander2004}. Under vaccination concern, certain mutant strains will finally have the competitive benefits amongst their contacts \cite{Scherer2002}. The impact of vaccination on the growth of strain contacts in multi-strain viruses has also been analyzed in many papers \cite{Gupta1997, May1994, McLean1995, Porco2000, Porco1998}. With the latest development of anti-COVID vaccines, numerous models have been proposed to provide insight into the impact of vaccination of a certain fraction of the populace on the dynamics of the COVID-19 pandemic. For instance, Fudolig et al. \cite{Fudolig2020} proposed a multi-strain model with vaccination for COVID-19 and studied the local stability of equilibria. Furthermore, there is some evidence of COVID-19's vaccine efficacy in Australia that mRNA vaccine has over 90 \% efficacy against COVID-19 infection, ChAdOx1 nCoV-19 has efficacy of 62\% against symptomatic infection in the intended two-dose schedule, BNT162b2 mRNA vaccine, which has 95\% efficacy against symptomatic infection \cite{MacIntyre2020}. Also, the Pfizer-BioNTech BNT162b2 mRNA vaccine has an efficacy over 95\%, Johnson \& Johnson [J\&J] Ad26 has an efficacy over 67\%, the AstraZeneca–Oxford ChAdOx1 nCov-19 vaccine has an efficacy over 67\%, and the Gamaleya GamCovidVac [Sputnik V] vaccine has the efficacy over 90\% \cite{Olliaro2021}. According to the Ministry of Health and Family Welfare, Government of India, the Indian vaccine also has vaccine efficacy over 70-90\% \cite{MOI}.

To determine effectual countermeasures, it is significant to develop mathematical models that support us in predicting and understanding the spread of COVID-19 and providing strategies on what could be implemented to limit its spread. Mathematical modeling in epidemiology provides a progressively greater room to public health research. This research discipline participates to sufficiently comprehend the studied epidemiological phenomenon and capture the distinct issues that can give rise to a terrible epidemic or even an alarming pandemic worldwide. The classical susceptible-infected-recovered (SIR) epidemic model was first proposed by Kermack et al. \cite{Kermack1927}. To obtain a sharper understanding of various vaccination strategies and their impacts on the number of infected individuals, Kermack-McKendrick type models have grabbed a vital role. This type of model has been benefited to comprehend vaccination dynamics on various diseases \cite{Alexander2004}. It is essential to emphasize that nowadays, Kermack-McKendrick kind mathematical models have helped explain COVID-19 epidemics properties worldwide. These models have been utilized to estimate the basic reproductive number associated with the disease and various parameters engaged in its spread. Additional use of this type of models has been focused on proposing and assessing the impact of different control measures categorized as NPIs. For example, the authors in papers \cite{Bugalia2020,Bugalia2021,Bajiya2020} proposed mathematical models for COVID-19 and analyzed the impact of NPIs on the disease dynamics. It is of paramount significance to develop mathematical models that can perfectly forecast the spread of COVID-19 so that the disease can be controlled and restrictions can be securely relaxed.  However, the infection incubation period may occupy a long time interval in some cases. An incubated person is not yet infectious in this time interval and remains latent. Therefore, another compartment of exposed individuals should be included in SIR, and the new model will have SEIR abbreviation \cite{Hethcote2000}.

Inspired by the evidence mentioned above about imperfect vaccine and mutation of the virus, in this paper we utilize an SEIR-type mathematical model to comprehend the dynamics of disease spread on the human populace under imperfect vaccination and two variants of the virus. 
The general methodology and mathematical model can be inferred to enhance the number of parameters and differential equations. We incorporate the vaccination compartment to the two-strain model to examine the effectiveness of the anti-COVID-19 vaccination, which is currently being employed in many countries to help battle the intense pandemic situation. It is supposed that the spread of a virus can mutate in the host to make a second, co-circulating, mutant strain. After some period of infection, the original strain, referred to as native strain, is converted to a mutant strain, such that a proportion of the people infected by the original strain are also carrying mutant strain, thus we consider mutation in our proposed epidemic model.
We intend to study the dynamical behavior of the strains' contacts under the vaccination scheme and investigate the impact of parameters (vaccination proportion, mutation rate, etc.) to demonstrate how they influence disease transmission. We also assume that recovered individuals from native strain have 100\% immunity against native strain but can be infected by mutant strain. The objectives of the present study are the following:
(i) Construction of an epidemic model that will describe the dynamics of mutant strain under imperfect vaccination. 
(ii) Investigation of the impact of an imperfect vaccine on the dynamics of the model. 
(iii) Investigation of different bifurcations with respect to variation of various parameters.
(iv) Application of the proposed model to the data of COVID-19 in India. 
(v) Observation of the COVID-19 dynamics with respect to the key parameter related to the mutation.

The remaining paper is organized as follows. Section \ref{dynmodel} describes the proposed model with imperfect vaccine and mutation. Section \ref{posandbond} discusses the non-negativity and boundedness of the solutions of the proposed system. Section \ref{dynanalysis} represents the dynamical analysis of the proposed system including the basic reproduction number, existence, and stability of possible equilibria, transcritical bifurcation, and Hopf bifurcation. Section \ref{Herd} describes the implications of disease control, and we obtained explicit expression of vaccine-induced herd immunity. Numerical evaluations have been presented in Section \ref{NumSim}. Cumulative cases and cumulative mortality data for COVID-19 pandemic in India have been used to parametrize the model, and the impact of different parameters on the cumulative cases and deaths have been shown in Section \ref{cal_model}. Sensitivity analysis of the parameters with respect to the basic reproduction number, cumulative cases, and cumulative deaths has been performed in Section \ref{Sensitivity}. The paper ends with a thorough discussion in Section \ref{Discussion}.

\section{Model formulation} \label{dynmodel}
We introduce a homogeneous two-strain model with imperfect vaccination. The system starts with a population exposed to both the native (original) and mutant (variant of original) strains of the virus. Some previous studies \cite{Gonzalez2021,DengX2021} revealed that the variant is more transmissible and severe than the original strain, and antibody neutralization is reduced in COVID-19 patients and vaccine recipients in various countries, including the US. Nature news has reported that the mutant strain is spreading quickly in India and has become the dominant strain \cite{Naturenews}. A mutation is accounted for in epidemic models through a term that transfers individuals infected with one of the strains into individuals infected with the other \cite{Liu2018,Cai2012,Martcheva2015}. We assume that vaccination is applied only to healthy individuals, so only susceptible individuals get vaccinated. Further, we also assume that the vaccine is imperfect, that is, the vaccinated individuals can become infected with both native and mutant strains of the virus. That happens at reduced transmission rates $\delta_1 \beta_1$ and $\delta_2 \beta_2,$ where $0 \leq \delta_1 \leq 1, 0 \leq \delta_2 \leq 1$ are the reduction coefficients of native and mutant strains, respectively. If $\delta_1 = 0,$ then vaccinated individuals will not get infected with native strain, i.e., the vaccine is perfect for native strain, and $\delta_1 = 1$ means vaccinated individuals get infected just like susceptible individuals, i.e., vaccine plays no protective role to native strain. Here $1-\delta_1$ describes the vaccine efficacy to native strain. A similar scenario applies to the mutant strain, and $1-\delta_2$ describes the vaccine efficacy to the mutant strain. Some vaccinated individuals can go back to susceptible individuals due to partial immunity. There are some shreds of evidence that there is a chance of a second COVID-19 infection after being diagnosed with first, from which recovered \cite{ScienceWire,ShastriJ2021,CDC}. Therefore, we assume that the recovered individuals of native strain can also become infected via mutant strain of the virus at a reduced rate $\delta_3 \beta_2$.

The model is composed of a system of differential equations that has eight compartments: 
susceptible compartment $S-$ individuals in this compartments are healthy but can be infected by both the native and mutant strains of the virus;
vaccinated compartment $V-$ individuals that applied to vaccination, these individuals can also become infected by both the native and mutant strains of the virus but at lower rates;
exposed compartment $E_1-$ individuals that are exposed to a native strain of virus;
exposed compartment $E_2-$ individuals that are exposed to a mutant strain of virus;
infected compartment $I_1-$ individuals that are infected to a native strain of virus;
infected compartment $I_2-$ Individuals that are infected to a mutant strain of virus;
recovered compartment $R_1-$ individuals that were infected to native strain and are now immune to the native strain but not immune to the mutant strain;
recovered compartment $R_2-$ individuals infected with mutant strain are now immune to both native and mutant strains and do not interact with the remaining compartments. The biological interpretations of the parameters involved in the model are given in Table \ref{parameters}. The schematic diagram of the model is given in Figure \ref{schematic}. The total population size is $N(t) = S(t)+V(t)+E_1(t)+E_2(t)+I_1(t)+I_2(t)+R_1(t)+R_2(t).$

\begin{table}[H]
	\caption{Biological interpretations of parameters.} \label{parameters}
	\begin{tabular}{p{2cm}p{11cm}}
		\hline
		Parameters	& Biological interpretations \\
		\hline
		$\Lambda$	& The recruitment rate at which new individuals enter in the susceptible population \\
		
		$\beta_1$	& Infection rate of the native strain \\
		
		$\beta_2$	& Infection rate of the mutant strain \\
		
		$p$	& Per capita vaccination rate of susceptible individuals \\
		
		$1/\mu$	& Average life expectancy of the individuals of all compartments \\
		
		$1- \delta_1$	& Efficacy of vaccine to native strain \\
		
		$1- \delta_2$	& Efficacy of vaccine to mutant strain \\
		
		$\gamma$	& Per capita rate of lost of immunity of vaccinated individuals \\
		
		$a_1$	& Per capita rate at which the exposed individuals of native strain become infectious \\
		
		$a_2$ &  Per capita rate at which the exposed individuals of mutant strain become infectious \\
		
		$\alpha_1$	& Per capita recovery rate of native strain \\
		
		$\alpha_2$	& Per capita recovery rate of mutant strain \\
		
		$d_1$	& Per capita death rate due to native strain \\
		
		$d_2$	& Per capita death rate due to mutant strain \\
		
		$\delta_3$	& Reduction coefficient of infection after recovery  \\
		
		$\nu_1$    & Per capita mutation rate of native strain \\
		
		\hline
	\end{tabular}
\end{table}

\begin{figure}[htp!]
	\includegraphics[scale=0.6]{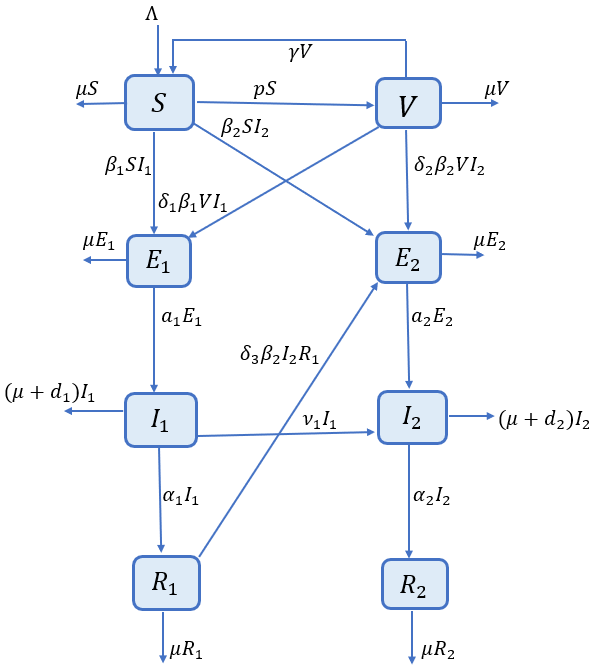}
	\caption{Schematic diagram of system \eqref{model}.}\label{schematic}
\end{figure}

Based on the parameters given in Table \ref{parameters} and schematic diagram \ref{schematic}, the dynamics of the disease transmission can be governed by the following system of ordinary differential equations:
\begin{equation}
\begin{aligned}\label{model}
\frac{dS}{dt} &= \Lambda - \beta_1 I_1 S - \beta_2 I_2 S - (\mu + p)S + \gamma V,\\
\frac{dV}{dt}& = pS - \delta_1 \beta_1 I_1 V - \delta_2 \beta_2 I_2 V - (\mu + \gamma) V,\\
\frac{dE_1}{dt} &= \beta_1(S + \delta_1 V) I_1 - (a_1 + \mu) E_1,\\
\frac{dE_2}{dt} &= \beta_2 (S  + \delta_2 V + \delta_3 R_1) I_2 - (a_2 + \mu) E_2,\\
\frac{dI_1}{dt} &= a_1 E_1  - (\alpha_1 + \mu + d_1 + \nu_1) I_1,\\
\frac{dI_2}{dt}& = a_2 E_2 - (\alpha_2 + \mu + d_2) I_2 + \nu_1 I_1,\\
\frac{dR_1}{dt}& = \alpha_1 I_1 - \delta_3 \beta_2 I_2 R_1 - \mu R_1,\\
\frac{dR_2}{dt}& = \alpha_2 I_2 - \mu R_2,
\end{aligned}
\end{equation}
with the initial conditions: $S(0)>0, V(0)\geq 0, E_1(0) \geq 0, E_2(0) \geq 0, I_1(0) \geq 0, I_2(0) \geq 0, R_1(0) \geq 0, R_2(0) \geq 0.$

\section{Well-posedness} \label{posandbond}
This section summarizes the positivity and boundedness of solutions of system \eqref{model}. Positivity is significant for biologically feasible solutions of the system while boundedness infers that solutions are finite. System \eqref{model} is given by the following bounded planes:
{\footnotesize \begin{equation*}
\begin{aligned}
\frac{dS}{dt}\Bigg|_{S=0,V\neq 0, E_1 \neq 0, E_2 \neq 0,I_1 \neq 0,I_2 \neq 0,R_1 \neq 0,R_2 \neq 0} &= \Lambda + \gamma V>0,\\
\frac{dV}{dt}\Bigg|_{S \neq 0,V= 0, E_1 \neq 0, E_2 \neq 0,I_1 \neq 0,I_2 \neq 0,R_1 \neq 0,R_2 \neq 0} &= pS \geq 0,\\
\frac{dE_1}{dt}\Bigg|_{S\neq 0,V\neq 0, E_1 = 0, E_2 \neq 0,I_1 \neq 0,I_2 \neq 0,R_1 \neq 0,R_2 \neq 0} &= \beta_1(S + \delta_1 V) I_1 \geq 0, \\
\frac{dE_2}{dt}\Bigg|_{S\neq 0,V\neq 0, E_1 \neq 0, E_2 = 0,I_1 \neq 0,I_2 \neq 0,R_1 \neq 0,R_2 \neq 0} &= \beta_2 (S  + \delta_2 V) I_2 + \delta_3 \beta_2 I_2 R_1 \geq 0,\\
\frac{dI_1}{dt}\Bigg|_{S\neq 0,V\neq 0, E_1 \neq 0, E_2 \neq 0,I_1 = 0,I_2 \neq 0,R_1 \neq 0,R_2 \neq 0} &= a_1 E_1 \geq 0,\\
\frac{dI_2}{dt}\Bigg|_{S\neq 0,V\neq 0, E_1 \neq 0, E_2 \neq 0,I_1 \neq 0,I_2 = 0,R_1 \neq 0,R_2 \neq 0} &= a_2 E_2 + \nu_1 I_1 \geq 0,\\
\frac{dR_1}{dt}\Bigg|_{S\neq 0,V\neq 0, E_1 \neq 0, E_2 \neq 0,I_1 \neq 0,I_2 \neq 0,R_1 = 0,R_2 \neq 0} &= \alpha_1 I_1 \geq 0, \\
\frac{dR_2}{dt}\Bigg|_{S\neq 0,V\neq 0, E_1 \neq 0, E_2 \neq 0,I_1 \neq 0,I_2 \neq 0,R_1 \neq 0,R_2 = 0} &= \alpha_2 I_2 \geq 0.
\end{aligned}
\end{equation*}}
Note that on each of the bounding planes of the non-negative cone of $\mathbb{R}^{8}_{+}$, all rates in the system \eqref{model} are non-negative. Thus, if we initiate this in the interior of this cone, we shall always remain in this cone as the direction of the vector field is inward on all the bounding planes. Therefore, the non-negativity of all solutions is guaranteed if we start from a non-negative initial point. Furthermore, system \eqref{model} also states that the population $N$ follows the below differential equation:
\begin{equation}
\frac{dN}{dt} = \Lambda - \mu N - d_1 I_1 - d_2 I_2,
\end{equation} 
which gives
\begin{equation*}
\Lambda - (\mu + d_1 + d_2) N \leq \frac{dN}{dt} \leq \Lambda - \mu N.
\end{equation*}
Now integrating the above inequality and using initial conditions, we obtain
\begin{equation*}
\frac{\Lambda}{\mu + d_1 + d_2} + \left(N(0) - \frac{\Lambda}{\mu + d_1 + d_2} \right) e^{-(\mu + d_1 + d_2)t} \leq N(t) \leq \frac{\Lambda}{\mu} + \left(N(0) - \frac{\Lambda}{\mu} \right) e^{-\mu t},
\end{equation*}
Considering $t \rightarrow +\infty,$ we obtain
\begin{equation*}
\frac{\Lambda}{\mu + d_1 + d_2} \leq \liminf \limits_{t \rightarrow +\infty}N(t) \leq \limsup \limits_{t \rightarrow +\infty}N(t) \leq \frac{\Lambda}{\mu}.
\end{equation*}
Hence, the feasible region for the system \eqref{model} is
\begin{equation*}
\Delta = \left\lbrace (S, V, E_1, E_2, I_1, I_2, R_1, R_2)\in \mathbb{R}^{8}_{+}: 0< S+V+E_1 +E_2 + I_1 + I_2 +R_1 +R_2 \leq \frac{\Lambda}{\mu} \subset \mathbb{R}^{8}_{+} \right\rbrace.
\end{equation*}
From the above analysis, we conclude the following consequence:
\begin{theorem}\label{pos_and_bound}
	The region $\Delta$ is positively invariant and attracting.
\end{theorem}
Therefore, the system \eqref{model} is well-posed and epidemiologically feasible since all variables remain nonnegative for all $t \geq 0$. Further, since the equations of system \eqref{model} are continuous and have continuous partial derivatives, then they satisfy the Lipschitz condition. Additionally, from Theorem \ref{pos_and_bound}, system \eqref{model} is uniformly bounded. Hence, the solution of the system \eqref{model} exists and is unique. 

\section{Rigorous analysis} \label{dynanalysis}
This section is devoted to investigating the dynamical behaviors of system \eqref{model} including the computation of the basic reproduction number, the existence of possible equilibria and their stability, and possible bifurcations. 
\subsection{Disease free equilibrium (DFE) and basic reproduction number}
%The equilibria are obtained by setting the equations of right hand side of the system equal to zero. 
The disease free equilibrium (DFE) can be obtained by setting all infected variables ($E_1, E_2, I_1, I_2$) equal to zero while all non-infected variables ($S, V, R_1, R_2$) are non-zero. The DFE of system \eqref{model} is given by
\begin{equation}\label{DFE}
D^0 = (S^0, V^0, E_1^0, E_2^0, I_1^0, I_2^0, R_1^0, R_2^0) = \left(\frac{\Lambda (\mu + \gamma)}{\mu (\mu + \gamma + p)}, \frac{\Lambda p}{\mu (\mu + \gamma + p)},0,0,0,0,0,0 \right).
\end{equation}
To obtain the basic reproduction number, we use the next generation method \cite{Diekmann2010, VandenDriessche2002}. By considering $x = (E_1, E_2, I_1, I_2)^{T},$ we have
\begin{equation}
x' = f(x)-v(x),
\end{equation}
where
\begin{align*}
f = \left(\begin{array}{c}
\beta_1(S + \delta_1 V) I_1 \\ 
\beta_2 (S  + \delta_2 V) I_2 + \delta_3 \beta_2 I_2 R_1 \\ 
0\\ 
0
\end{array}  \right), \quad v = \left( \begin{array}{c}
(a_1 +\mu) E_1\\ 
(a_2 +\mu) E_2\\ 
(\alpha_1 + \mu +d_1 + \nu_1)I_1 - a_1 E_1\\ 
(\alpha_2 + \mu +d_2 )I_2 - \nu_1 I_1 - a_2 E_2
\end{array}  \right). 
\end{align*}
The Jacobian of $f(x)$ and $v(x)$ at $D^0$ are
\begin{align*}
F = Df(D^0) = \left( \begin{array}{cccc}
0& 0  & \beta_1(S^0 + \delta_1 V^0)  & 0  \\ 
0& 0 & 0  & \beta_2(S^0 + \delta_2 V^0)  \\ 
0& 0 & 0 & 0 \\ 
0& 0 & 0 & 0
\end{array} \right) ,
\end{align*}
and 
\begin{align*}
V &= Dv(D^0) = \left(\begin{array}{cccc}
(a_1 + \mu)& 0 & 0 & 0 \\ 
0 & (a_2 + \mu) & 0 & 0 \\ 
-a_1& 0 & (\alpha_1 + \mu + d_1 + \nu_1) & 0 \\ 
0 & -a_2 & -\nu_1 &  (\alpha_2 + \mu + d_2)
\end{array}  \right), 
\end{align*} 
respectively. Hence, the basic reproduction number for system \eqref{model} is the dominant eigenvalue or spectral radius of the next generation matrix $FV^{-1}$ which is given by the following expression:
\begin{align}
R_0 = \rho(FV^{-1}) = \max \left\lbrace R_1, R_2 \right\rbrace, 
\end{align}
where 
\begin{align*}
R_1 & = \frac{ \beta_1 a_1 (S ^0 + \delta_1 V^0)}{(\mu + a_1)(\mu + d_1 + \alpha_1 + \nu_1)}\\
& = \frac{\Lambda \beta_1 a_1 (\gamma + \mu + p\delta_1)}{\mu (\mu + a_1)(\mu + d_1 + \alpha_1 + \nu_1)(p+\gamma + \mu)}, \\
R_2 & = \frac{ \beta_2 a_2 (S ^0 + \delta_2 V^0)}{(\mu + a_2)(\mu + d_2 + \alpha_2)} \\
&= \frac{\Lambda \beta_2 a_2 (\gamma + \mu + p\delta_2)}{\mu (\mu + a_2)(\mu + d_2 + \alpha_2 )(p+\gamma + \mu)}.
\end{align*}
Here $R_1 (R_2)$ represents the average number of secondary infection cases generated by a single infectious individuals of the native (mutant) strain of the virus, called the basic reproduction number of the native (mutant) strain.

\subsubsection{Interpretation of the basic reproduction number}
As stated above, the basic reproduction number $R_0$ is the maximum of the two basic reproduction numbers, $R_1$ and $R_2$. The basic reproduction number $R_1$ is given by the product of the infection rate of the susceptible (unvaccinated) and vaccinated individuals by native strain infectious individuals (near the disease-free equilibrium), [$\beta_1 (S^0 + \delta_1 V^0)$], the proportion of the exposed individuals to the native strain that survived in the exposed class ($E_1$) and moved to infected compartment ($I_1$) [$\frac{a_1}{\mu + a_1}$], and the average time duration in the  infectious class ($I_1$) [$\frac{1}{\mu + d_1 + \alpha_1 + \nu_1}$]. The basic reproduction number $R_2$ is given by the product of the infection rate of the susceptible (unvaccinated) and vaccinated individuals by the mutant strain infectious individuals (near the disease-free equilibrium), [$\beta_2 (S^0 + \delta_1 V^0)$], the proportion of the exposed individuals to the mutant strain that survived in the exposed class ($E_2$) and moved to infected compartment ($I_2$) [$\frac{a_2}{\mu + a_2}$], and the average time duration in the infectious class ($I_2$) [$\frac{1}{\mu + d_2 + \alpha_2 }$]. 

\begin{remark}
	If $\delta_1 = 0, \delta_2 = 0,$ then vaccine is perfectly effective to both native and mutant strains and if $p=0$ then the system reduces without vaccination. For this case, the basic reproduction numbers of the native strain and mutant strain, $R_{1wv}$ and $R_{2wv}$ are given by, respectively, 
	\begin{align*}
	R_{1wv} & = \frac{\Lambda \beta_1 a_1}{\mu (\mu + a_1)(\mu + d_1 + \alpha_1 + \nu_1)}, \\
	R_{2wv} & = \frac{\Lambda \beta_2 a_2}{\mu (\mu + a_2)(\mu + d_2 + \alpha_2 )}.
	\end{align*}
	Thus, the basic reproduction number of system without vaccination is given by $R_{0wv} = \max \left\lbrace R_{1wv}, R_{2wv} \right\rbrace.$ 
\end{remark}

\begin{theorem}\label{dfe_local}
	If $R_0 = \max \left\lbrace R_1, R_2 \right\rbrace <1,$ the DFE $(D^0)$ is locally asymptotically stable; if $R_0 = \max \left\lbrace R_1, R_2 \right\rbrace >1,$ $D^0$ is unstable.
\end{theorem}
\begin{proof}
	By linearizing the system \eqref{model} at $D^0$, we obtain the following characteristic equation 
	\begin{equation}\label{chareq}
	\begin{aligned}
	(\lambda + \mu)^2 (\lambda + p + \gamma + \mu) (\lambda^2 + \lambda (2\mu + a_1 + d_1 + \alpha_1) + (\mu + a_1)(\mu + d_1 + \alpha_1 + \nu_1) (1 - R_1)) \\
	(\lambda^2 + \lambda (2\mu + a_2 + d_2 + \alpha_2) + (\mu + a_2)(\mu + d_2 + \alpha_2) (1 - R_2)) = 0.
	\end{aligned}
	\end{equation}
	From the above characteristic equation, it is easy to see that all the roots of Eq. \eqref{chareq} are negative or have negative real parts for $R_0 < 1.$ Hence, the $D^0$ of system \eqref{model} is locally asymptotically stable for $R_0 < 1.$ If $R_0 > 1,$ at least one of the roots of Eq. \eqref{chareq} has positive real part. Hence, $D^0$ is unstable for $R_0 > 1.$ 
\end{proof} 

\begin{theorem}\label{dfe_gloal}
 $D^0$ is globally asymptotically stable, whenever $R_0 = \max \left\lbrace R_1, R_2 \right\rbrace <1$.
\end{theorem} 
\begin{proof}
		To prove the global stability of $D^0$, we follow the approach given by Castillo-Chavez et al. \cite{Castillo-Chavez}. We rewrite the system \eqref{model} as follows
	\begin{equation}
	\begin{aligned}
	\frac{dX}{dt} &= F(X,Y),\\
	\frac{dY}{dt} &= G(X,Y), \quad G(X,0)=0,
	\end{aligned}
	\end{equation}
	where $X = (S,V,R_1,R_2)\in \mathbb{R}^4$ signifies the number of uninfected individuals and $Y = (E_1,E_2,I_1,I_2) \in \mathbb{R}^4 $ signifies the number of infected individuals. Disease-free equilibrium $(D^0)$ is globally stable if the following two conditions are fulfilled:
	\begin{itemize}
		\item[(H1)] For $\frac{dX}{dt} = F(X,Y),$ $X^*$ is globally asymptotically stable,
		\item[(H2)] $G(X,Y) = MY - \hat{G}(X,Y),~ \hat{G}(X,Y)>0$ for $(X,Y)\in \Delta$,
	\end{itemize}
	where $M = D_Y G(X^*, 0)$ is an $M$-matrix. For the system \eqref{model}, we have 
	\begin{equation}\label{model_9}
	F(X,0) = \left(\begin{array}{c}
	\Lambda - (\mu +p) S + \gamma V\\ 
	pS - (\mu + \gamma)V \\
	0 \\
	0
	\end{array}  \right). 
	\end{equation}
		It is obvious that the equilibrium $X^* = \left(\frac{\Lambda (\mu + \gamma)}{\mu (\mu + \gamma + p)}, \frac{\Lambda p}{\mu (\mu + \gamma + p)}, 0, 0 \right) $ is globally asymptotically stable of system \eqref{model_9}. Further, for system \eqref{model}, we obtain 
	\begin{align*}
	M &= \left(\begin{array}{cccc}
	-(\mu +a_1) & 0 & \beta_1 (S^0 + \delta_1 V^0) & 0 \\ 
	0 & -(\mu +a_2) & 0 & \beta_2 (S^0 + \delta_2 V^0) \\
	a_1 & 0 & -(\alpha_1 + \mu + d_1 + \nu_1) & 0 \\
	0 & a_2 & \nu_1 & -(\alpha_2 + \mu + d_2)
	\end{array} \right), \\
	\hat{G}(X,Y) &= \left(\begin{array}{c}
	\beta_1 I_1((S^0 + \delta_1 V^0)-(S+\delta_1 V) ) \\ 
	\beta_2 I_2((S^0 + \delta_2 V^0)-(S+\delta_2 V + \delta_3 R_1) ) \\
	0 \\
	0
	\end{array}  \right). 
	\end{align*}
	It is clear that $\hat{G}(X,Y)\geq 0$. Hence, $D^0$ is globally stable, i.e. every solution of the system \eqref{model} approaches the DFE ($D^0$) as $t \rightarrow \infty$ for $R_0 <1.$ Thus, the disease (i.e., both native and mutant strains) will be eliminated from the community if $R_0 <1.$ Consequently, $R_0$ represents the threshold value for the existence of other positive equilibria of the system \eqref{model}. Moreover, it is eminent that the basic reproduction number $(R_0)$ represents the average number of secondary infections that occurred from a single infected individual in the whole susceptible population. Therefore, if $R_0 < 1$, each infected individual in the entire infectious period will produce less than one infected individual on average, which implies that the disease will die out. However, if $R_0 > 1,$ then each infected individual in the whole infectious period having contact with susceptible individuals will generate more than one infected individual; this leads to the disease invading the susceptible population.     
\end{proof}

It must be mentioned that, for mathematical (endemic) models such as \eqref{model}, the epidemiological necessity $R_0 <1$ is sufficient as well as necessary for eradication of the disease. This is because, for such mathematical models (i.e., Kermack-McKendrick models with demographic dynamics), the disease will persist whenever $R_0 >1$ (this is because the pool of new susceptible individuals will continuously be refilled, by immigration or birth, thereby letting the disease to maintain itself in the community). If the demographic effects are not allowed (i.e., in the case of a single outbreak/epidemic model is used), the epidemiological condition $R_0 <1$ is only sufficient, but not necessary, for eradicating the epidemic. For such epidemic models (with no demographic dynamics), the disease always dies out with time (irrespective of the value of the basic reproduction number of the epidemic models). In other words, even if the basic reproduction number exceeds unity, the disease will eventually die out; this is because the endemic rises and reach a peak.

\subsection{Mutant dominant equilibrium and its stability}
First of all, it should be mentioned that in the absence of native strain ($I_1 = 0$), the system \eqref{model} reduces to the following subsystem:

 \begin{equation}\label{mutantmodel}
 \begin{aligned}
 \frac{dS}{dt} &= \Lambda - \beta_2 I_2 S - (\mu + p)S + \gamma V,\\
 \frac{dV}{dt}& = pS - \delta_2 \beta_2 I_2 V - (\mu + \gamma) V,\\
 \frac{dE_2}{dt} &= \beta_2 (S  + \delta_2 V) I_2 - (a_2 + \mu) E_2,\\
 \frac{dI_2}{dt}& = a_2 E_2 - (\alpha_2 + \mu + d_2) I_2,\\
 \frac{dR_2}{dt}& = \alpha_2 I_2 - \mu R_2.
 \end{aligned}
 \end{equation}
 The analysis of the sub-system \eqref{mutantmodel} will be considered in the following positively invariant region 
  \begin{equation*}
  \Delta_{I_2} = \left\lbrace (S, V, E_2, I_2, R_2)\in \mathbb{R}^{5}_{+}: 0< S+V +E_2 + I_2 +R_2 \leq \frac{\Lambda}{\mu} \subset \mathbb{R}^{5}_{+} \right\rbrace. 
  \end{equation*}
The mutant dominant equilibrium is given by $D^2 = (S^2, V^2, 0, E_2^2, 0, I_2^2, 0, R_2^2)$, where the components of $D^2$ can be obtained by solving the equations of right hand side of the subsystem \eqref{mutantmodel}. Thus, 
\begin{align*}
S^2 & = \frac{\Lambda}{\mu} - \frac{p\left\lbrace \Lambda a_2 - (a_2 + \mu)(\alpha_2 + \mu +d_2)I_2^2 \right\rbrace }{\mu a_2 (\delta_2 \beta_2 I_2^2 + \mu + \gamma + p)} - \frac{(a_2 + \mu)(\alpha_2 + \mu + d_2)I_2^2}{\mu a_2},\\
V^2 & = \frac{p\left\lbrace \Lambda a_2 - (a_2 + \mu)(\alpha_2 + \mu +d_2)I_2^2 \right\rbrace }{\mu a_2 (\delta_2 \beta_2 I_2^2 + \mu + \gamma + p)}, \quad E_2^2 = \frac{(\alpha_2 + \mu +d_2)I_2^2}{a_2}, \quad R_2^2 = \frac{\alpha_2 I_2^2}{\mu},
\end{align*}
and $I_2^2$ satisfies the following equation:
\begin{align}\label{MSE}
k_1^{'}I_2^2 + k_2^{'}I_2 + k_3^{'} = 0,
\end{align}
where 
\begin{align*}
k_1^{'} &= \beta_2^2 (\mu + a_2)(\mu +d_2 + \alpha_2)\delta_2,\\
k_2^{'} &= \beta_2((\mu+a_2)(\mu +d_2 + \alpha_2)(\gamma + \mu + (p+\mu)\delta_2)-\Lambda \beta_2 a_2 \delta_2),\\
k_3^{'} &= \mu (\mu + a_2)(\mu +d_2 + \alpha_2)(\gamma + \mu + p)(1-R_2).
\end{align*}
We can see that Eq. \eqref{MSE} has zero, one, or two roots, depending on parameter values. For the case $0<\delta_2 \leq 1,$ $k_3^{'} <0$ if $R_2 > 1$, and  $k_3^{'} > 0$ if $R_2 < 1$. Since Eq. \eqref{MSE} is quadratic equation, therefore if $R_2 > 1$, then Eq. \eqref{MSE} has a unique positive root and there is a unique mutant dominant equilibrium. If $R_2 = 1$, then $k_3^{'} = 0$ and there is unique non-zero solution of \eqref{MSE}, given by $I_2^2 = -\frac{k_2^{'}}{k_1^{'}},$ which is positive if and only if $k_2^{'}<0.$ If $R_2 = 1, k_3^{'} = 0,$ then
\begin{equation}\label{eq_7}
\Lambda \beta_2 a_2 (\gamma + \mu + p\delta_2)=\mu (\mu + a_2)(\mu + d_2 + \alpha_2)(p+\gamma + \mu).
\end{equation}
The condition $k_2^{'} < 0$ gives
\begin{equation*}
(\mu+a_2)(\mu +d_2 + \alpha_2)(\gamma + \mu + (p+\mu)\delta_2) < \Lambda \beta_2 a_2 \delta_2,
\end{equation*}
combined with \eqref{eq_7}, we obtain
\begin{equation*}
(\gamma + \mu)^2 + (p\delta_2)^2 + \mu p \delta_2^2 + 2\gamma p \delta_2 + \mu p \delta_2 <0,
\end{equation*}
which is not possible. Hence, if $R_2 \leq 1$, system \eqref{model} has no mutant dominant equilibrium. Furthermore, it should be stated that for the equilibrium $D^2$ to exist, it is necessary that the native strain dies out asymptotically (i.e., $R_1 \leq 1$). Thus, we conclude that system \eqref{mutantmodel} has a unique mutant dominant equilibrium ($D^2$) whenever $R_2 >1$ and $R_1 \leq 1.$ Further, for the stability of the equilibrium $D^2$, we have the following result:
\begin{theorem}
	The unique mutant dominant equilibrium ($D^2$) is globally asymptotically stable whenever $R_2>1$ and $R_1 \leq 1.$
\end{theorem}
\begin{proof}
	We consider
	\begin{align*}
	x=\frac{S}{S^2},~ y=\frac{V}{V^2},~ z=\frac{E_2}{E_2^2},~ u=\frac{I_2}{I_2^2},
	\end{align*}
	and with the help of Eqs. of right hand side of the system \eqref{mutantmodel}, the system \eqref{mutantmodel} can be rewritten as follows:
	\begin{equation}\label{modellyp1}
	\begin{aligned}
	x' &= x\Big[\frac{\Lambda}{S^2}\Big(\frac{1}{x}-1\Big) - \beta_2 I_2^2 (u-1) + \frac{\gamma V^2}{S^2}\Big(\frac{y}{x} -1 \Big) \Big],\\
	y' & =  y\Big[\frac{pS^2}{V^2}\Big(\frac{x}{y}-1\Big) - \delta_2 \beta_2 I_2^2 (u-1) \Big],\\
	z' &= z \frac{\beta_2 I_2^2}{E_2^2}\Big[S^2 \Big(\frac{xu}{z}-1 \Big) + \delta_2 V^2 \Big(\frac{yu}{z}-1 \Big) \Big],\\
	u' &= u  \frac{a_2 E_2^2}{I_2^2}\Big[\frac{z}{u}-1 \Big].
	\end{aligned}
	\end{equation}
	Further, we consider the following Lyapunov function
	\begin{equation*}
	\begin{aligned}
	Z = k_1 S^2 (x-1-\ln x) + k_2 V^2 (y-1-\ln y) + k_3 E_2^2 (z-1-\ln z) + k_4 I_2^2 (u-1-\ln u),
	\end{aligned}
	\end{equation*}
	where the positive constants $k_1, k_2, k_3,$ and $k_4$ will be defined below. By differentiating $Z$ with respect to $t$ along the solutions of \eqref{mutantmodel}, we obtain
	
	\begin{equation*}
	\begin{aligned}
	Z' =& k_1 (x-1) \Big[\Lambda \Big(\frac{1}{x} - 1\Big) - \beta_2 S^2 I_2^2 (u-1) + \gamma V^2 \Big(\frac{y}{x} -1 \Big) \Big] \\
	&+ k_2 (y-1) \Big[ps^2 \Big(\frac{x}{y} - 1\Big) - \delta_2 \beta_2 V^2 I_2^2 (u-1) \Big] \\
	&+ k_3 \beta_2 I_2^2 (z-1) \Big[S^2 \Big(\frac{xu}{z}-1\Big) + \delta_2 V^2 \Big(\frac{yu}{z}-1\Big) \Big] + k_4 a_2 E_2^2 (u-1) \Big(\frac{z}{u}-1\Big) \\
	=& k_1 (2\Lambda + \gamma V^2 - \beta_2 S^2 I_2^2) + k_2 (pS^2 - \delta_2 \beta_2 V^2 I_2^2) + k_3 (\beta_2 S^2 I_2^2 + \delta_2 \beta_2 V^2 I_2^2) \\
	& + k_4 a_2 E_2^2 - (k_1 \Lambda + k_1 \gamma V^2 - k_1 \beta_2 S^2 I_2^2 - k_2 p S^2)x - k_1 \Lambda \frac{1}{x}\\
	& - (-k_1 \gamma V^2 + k_2 p S^2 - k_2 \delta_2 \beta_2 V^2 I_2^2)y - k_1 \gamma V^2 \frac{y}{x} - k_2 p S^2 \frac{x}{y} \\
	&- (k_1 \beta_2 S^2 I_2^2 - k_3 \beta_2 S^2 I_2^2)ux - (-k_1 \beta_2 S^2 I_2^2 - k_2 \delta_2 \beta_2 V^2 I_2^2 + k_4 a_2 E_2^2)u \\
	&- (k_2 \delta_2 \beta_2 V^2 I_2^2 - k_3 \delta_2 \beta_2 V^2 I_2^2)yu - (k_3 \beta_2 S^2 I_2^2 + k_3 \delta_2 \beta_2 V^2 I_2^2 - k_4 a_2 E_2^2)z \\
	&- k_3 \beta_2 S^2 I_2^2 \frac{xu}{z} - k_3 \delta_2 \beta_2 V^2 I_2^2 \frac{yu}{z} - k_4 a_2 E_2^2 \frac{z}{u} \\
	=: & G(x,y,z,u).
	\end{aligned}
	\end{equation*}
	Choose the positive constants $k_1, k_2, k_3,$ and $k_4$ as follows:
	\begin{equation*}
	k_1 = k_2 = k_3 =1, k_4 = \frac{a_2 + \mu}{a_2}.
	\end{equation*} 
	Substituting the above values into the function $G(x,y,z,u)$ gives
	\begin{equation*}
	\begin{aligned}
	G(x,y,z,u)  =& 2 \Lambda + \gamma V^2 + pS^2 + (a_2 + \mu) E_2^2 - \mu S^2 x - \Lambda \frac{1}{x} - \mu V^2 y - \gamma V^2 \frac{y}{x} \\
	& - pS^2 \frac{x}{y} - \beta_2 S^2 I_2^2 \frac{xu}{z} - \delta_2 \beta_2 V^2 I_2^2 \frac{yu}{z} - (a_2 + \mu) E_2^2 \frac{z}{u} \\ 
	=& \mu S^2 \Big(2-x-\frac{1}{x}\Big) + \gamma V^2 \Big(2- \frac{x}{y} - \frac{y}{x}\Big) + \mu V^2 \Big(3- \frac{1}{x} - y - \frac{x}{y}\Big) \\
	& + \beta_2 S^2 I_2^2 \Big(3- \frac{1}{x} - \frac{xu}{z} -\frac{z}{u}\Big) + \delta_2 \beta_2 V^2 I_2^2 \Big(4- \frac{1}{x} - \frac{x}{y} - \frac{yu}{z} - \frac{z}{u}\Big).
	\end{aligned}
	\end{equation*} 
	By the property that the arithmetic mean is greater than or equal to the geometric mean, $G(x,y,z,u) \leq 0,$ and the equality holds only for $x=y=1$ and $z=u,$ i.e., $$\left\lbrace (x,y,z,u) \in \Delta_{I_2} : G(x,y,z,u) = 0 \right\rbrace \equiv \left\lbrace (x,y,z,u): x=y=1,z=u \right\rbrace, $$ which corresponds to the set
	$$\Delta_{I_2}' = \left\lbrace (S,V, E_2, I_2): S = S^*, V=V^*, \frac{E_2}{E_2^*}=\frac{I_2}{I_2^*} \right\rbrace \subset \Delta_{I_2} \subset \Delta.$$ It is evident to see that the maximum invariant set of \eqref{mutantmodel} on the set $\Delta_{I_2}'$ is the singleton $\left\lbrace D^2 \right\rbrace $, then the mutant dominant equilibrium $D^2$ is globally stable in $\Delta_{I_2} \subset \Delta$ by LaSalle's Invariance Principle \cite{LaSalle1976}.
\end{proof}

\subsection{Coexistence equilibrium and its stability}
This section examines the existence and global stability of the coexistence equilibrium of system \eqref{model}. First of all, let us consider the endemic equilibrium $D^* = (S^*, V^*, E_1^*, E_2^*, I_1^*, I_2^*, R_1^*, R_2^*)$, then $S^*, V^*,$ $E_1^*,E_2^*, I_1^*, I_2^*, R_1^*,$ and $R_2^*$ satisfy the following equations: 
\begin{equation}\label{eqn_endemic}
\begin{aligned}
&\Lambda - \beta_1 I_1 S - \beta_2 I_2 S - (\mu + p)S + \gamma V=0,\\
&pS - \delta_1 \beta_1 I_1 V - \delta_2 \beta_2 I_2 V - (\mu + \gamma) V=0,\\
& \beta_1(S + \delta_1 V) I_1 - (a_1 + \mu) E_1=0,\\
& \beta_2 (S  + \delta_2 V + \delta_3 R_1) I_2 - (a_2 + \mu) E_2=0,\\
& a_1 E_1  - (\alpha_1 + \mu + d_1 + \nu_1) I_1=0,\\
& a_2 E_2 - (\alpha_2 + \mu + d_2) I_2 + \nu_1 I_1 =0,\\
&\alpha_1 I_1 - \delta_3 \beta_2 I_2 R_1 - \mu R_1 = 0,\\
& \alpha_2 I_2 - \mu R_2 = 0.
\end{aligned}
\end{equation}
The above Eqs \eqref{eqn_endemic} lead the following expressions:
\begin{align*}
S^* &= \frac{\Lambda (\beta_1 \delta_1 I_1^* + \beta_2 \delta_2 I_2^* +\mu + \gamma)}{(\beta_1 I_1^* + \beta_2 I_2^* +\mu + p)(\beta_1 \delta_1 I_1^* + \beta_2 \delta_2 I_2^* +\mu + \gamma)-p\gamma}, \\
V^* &= \frac{\Lambda p}{(\beta_1 I_1^* + \beta_2 I_2^* +\mu + p)(\beta_1 \delta_1 I_1^* + \beta_2 \delta_2 I_2^* +\mu + \gamma)-p\gamma}, \\
E_1^* &= \frac{(\alpha_1 + \mu +d_1 +\nu_1)I_1^*}{a_1}, \quad E_2^* = \frac{(\alpha_2 + \mu +d_2)I_2^* - \nu_1 I_1^*}{a_2}, \\
R_1^* &= \frac{\alpha_1 I_1^*}{\delta_2 \beta_2 I_2^* + \mu}, \quad  R_2^* = \frac{\alpha_2 I_2^*}{\mu},
\end{align*}
and $I_1^*, I_2^*$ are the solutions of the following equations:
{\footnotesize \begin{equation} \label{end_sys}
\begin{aligned}
F(I_1, I_2) \equiv & \beta_2^2 \delta_2 I_2^2 + \beta_2 \Big(-\frac{\Lambda \beta_1 a_1 \delta_2}{(\mu + a_1)(\mu + d_1 + \alpha_1 + \nu_1)} + \gamma + \mu + (p+\mu)\delta_2 + \beta_1 (\delta_1 + \delta_2)I_1 \Big) I_2 \\
&+ \beta_1^2 \delta_1 I_1^2 + \beta_1 \Big(-\frac{\Lambda \beta_1 a_1 \delta_1}{(\mu + a_1)(\mu + d_1 + \alpha_1 + \nu_1)} + \gamma + \mu + (p+\mu)\delta_1 \Big) I_1 \\
 &+ \mu (p+\mu + \gamma)(1-R_1)=0,\\
 G(I_1, I_2) \equiv & I_2^3 \Bigg( \beta_2^2 \delta_2 \Big(1 -\frac{\Lambda a_2 \beta_2 \delta_2}{(\mu + a_2)(\mu + d_2 + \alpha_2)} - \frac{a_2 \alpha_1 \beta_2 \delta_3 I_1}{(\mu + a_2)(\mu + d_2 + \alpha_2)} \Big)  \Bigg) \\
 & + I_2^2 \Bigg(\beta_2 (\gamma + \mu + \delta_2(p+\mu)) - \frac{a_2 \beta_2^2 \delta_2 \Lambda (\gamma + 2\mu - p\delta_2)}{(\mu + a_2)(\mu + d_2 + \alpha_2)} - \frac{I_1^2 a_2 \alpha_1 \beta_1 \beta_2^2 \delta_3 (\delta_1 + \delta_2)}{(\mu + a_2)(\mu + d_2 + \alpha_2)}\\
 &+ I_1 \Big(\beta_1 \beta_2(\delta_1 + \delta_2) - \frac{\beta_2^2 (a_2 \alpha_1 \delta_3 (\gamma + \mu + (p+\mu)\delta_2) + \Lambda a_2 \beta_1 \delta_1 \delta_2 + (\mu + a_2)\delta_1 \nu_1)}{(\mu + a_2)(\mu + d_2 + \alpha_2)} \Big) \Bigg) \\
 & + I_2 \Bigg(\mu (p+\gamma + \mu) - \frac{\Lambda \beta_2 a_2 \mu (\gamma + \mu + p\delta_2)}{(\mu + a_2)(\mu + d_2 + \alpha_2)} - \frac{a_2 \alpha_1 \beta_1^2 \beta_2 \delta_1 \delta_3 I_1^3}{(\mu + a_2)(\mu + d_2 + \alpha_2)} \\
 &+ I_1^2 \Big(1- \frac{\beta_1 \beta_2 \delta_3 \alpha_1 a_2 (\gamma + \mu +  (p+\mu)\delta_1)}{(\mu + a_2)(\mu + d_2 + \alpha_2)} + \frac{(\delta_1 + \delta_2)\nu_1}{(\mu + d_2 + \alpha_2)} \Big) \\
 &+ I_1 \Big( \frac{\nu_1(\gamma + \mu +  (p+\mu)\delta_2) - \beta_2 \mu a_2 \alpha_1 \delta_3 (\gamma + \mu + p)-\Lambda \mu a_2 \beta_1 \beta_2 \delta_1}{(\mu + a_2)(\mu + d_2 + \alpha_2)}\\
 &+ \beta_1(\gamma + \mu +  (p+\mu)\delta_1) \Big) \Bigg) - \frac{I_1}{(\mu + d_2 + \alpha_2)} \Bigg(\nu_1 \mu (p+\gamma + \mu) \\
 &+ I_1 \beta_1 \nu_1 (\gamma + \mu +  (p+\mu)\delta_1) + I_1^2 \beta_1^2 \delta_1 \nu_1 \Bigg)=0.
\end{aligned}
\end{equation}}
If the system \eqref{end_sys} admits a solution, then the system \eqref{model} will have an endemic equilibrium. Obtaining the explicit expression for the exact solution of the non-linear autonomous system \eqref{end_sys} is a daunting task. However, we prove the global stability of the endemic equilibrium for the special case ($\delta_3 = 0$ and $\nu_1 = 0$) of system \eqref{model} in the subsequent theorem:

\begin{theorem}
	If the endemic equilibrium $D^* = (S^*, V^*, E_1^*, E_2^*, I_1^*, I_2^*, R_1^*, R_2^*)$ exists for $\delta_3 = 0$ and $\nu_1 = 0,$ then it is globally asymptotically stable.
\end{theorem}
\begin{proof}
	For $\delta_3 = 0$ and $\nu_1 = 0,$ the endemic equilibrium $D^* = (S^*, V^*, E_1^*, E_2^*, I_1^*, I_2^*),$ $S^*, V^*,$ $E_1^*,E_2^*,$ $I_1^*,$ and $I_2^*$ satisfy the following equations:
	\begin{equation}\label{eqn_endemic1}
	\begin{aligned}
	 &\Lambda - \beta_1 I_1 S - \beta_2 I_2 S - (\mu + p)S + \gamma V=0,\\
	&pS - \delta_1 \beta_1 I_1 V - \delta_2 \beta_2 I_2 V - (\mu + \gamma) V=0,\\
	& \beta_1(S + \delta_1 V) I_1 - (a_1 + \mu) E_1=0,\\
	& \beta_2 (S  + \delta_2 V) I_2 - (a_2 + \mu) E_2=0,\\
	& a_1 E_1  - (\alpha_1 + \mu + d_1) I_1=0,\\
	& a_2 E_2 - (\alpha_2 + \mu + d_2) I_2 =0.
	\end{aligned}
	\end{equation}
	Since $R_1$ and $R_2$ do not appear explicitly in the first six equations, therefore we omit them. By applying Eq. \eqref{eqn_endemic1} and denoting 
	\begin{equation*}
	x=\frac{S}{S^*},~ y=\frac{V}{V^*},~ z=\frac{E_1}{E_1^*},~ u=\frac{E_2}{E_2^*},~ v=\frac{I_1}{I_1^*},~ w=\frac{I_2}{I_2^*}, 
	\end{equation*}
	the system \eqref{model} can be rewritten as follows:
	\begin{equation}\label{modellyp}
	\begin{aligned}
	x' &= x\Big[\frac{\Lambda}{S^*}\Big(\frac{1}{x}-1\Big) - \beta_1 I_1^* (v-1) - \beta_2 I_2^* (w-1) + \frac{\gamma V^*}{S^*}\Big(\frac{y}{x} -1 \Big) \Big],\\
	y' & =  y\Big[\frac{pS^*}{V^*}\Big(\frac{x}{y}-1\Big) - \delta_1 \beta_1 I_1^* (v-1) - \delta_2 \beta_2 I_2^* (w-1) \Big],\\
	z' &= z \frac{\beta_1 I_1^*}{E_1^*}\Big[S^* \Big(\frac{xv}{z}-1 \Big) + \delta_1 V^* \Big(\frac{yv}{z}-1 \Big) \Big],\\
	u' &= u \frac{\beta_2 I_2^*}{E_2^*}\Big[S^* \Big(\frac{xw}{u}-1 \Big) + \delta_2 V^* \Big(\frac{yw}{u}-1 \Big) \Big],\\
	v' &= v  \frac{a_1 E_1^*}{I_1^*}\Big[\frac{z}{v}-1 \Big],\\
	w' &= w  \frac{a_2 E_2^*}{I_2^*}\Big[\frac{u}{w}-1 \Big].
	\end{aligned}
	\end{equation}
	It is clear that the endemic equilibrium $D^*$ of \eqref{model} corresponds to the positive equilibrium $\bar{D}^*(1,1,1,1,1,1)$ of \eqref{modellyp}, and that the global stability of $\bar{D}^*$ is same as that of $D^*,$ therefore,  we will discuss the global stability of the equilibrium $\bar{D}^*$ of system \eqref{modellyp} instead of $D^*.$
	
	Define the Lyapunov function
	\begin{equation*}
	\begin{aligned}
	L =& k_1 S^* (x-1-\ln x) + k_2 V^* (y-1-\ln y) + k_3 E_1^* (z-1-\ln z) + k_4 E_2^* (u-1-\ln u) \\
	&+ k_5 I_1^* (v-1-\ln v) + k_6 I_2^* (w-1-\ln w),
	\end{aligned}
	\end{equation*}
	where the positive numbers $k_1, k_2, k_3, k_4, k_5,$ and $k_6$ will be given below, then differentiating $L$ with respect to $t$ along solutions of \eqref{modellyp}, we obtain
	\begin{equation*}
	\begin{aligned}
	L' =& k_1 (x-1) \Big[\Lambda \Big(\frac{1}{x}-1 \Big) - \beta_1 I_1^* S^* (v-1) - \beta_2 I_2^* S^* (w-1) +\gamma V^*\Big(\frac{y}{x} -1 \Big) \Big] \\
	&+ k_2 (y-1)\Big[pS^*\Big(\frac{x}{y}-1\Big) - \delta_1 \beta_1 I_1^* V^* (v-1) - \delta_2 \beta_2 I_2^* V^* (w-1) \Big]\\
	&+ k_3 (z-1) \beta_1 I_1^* \Big[S^* \Big(\frac{xv}{z}-1 \Big) + \delta_1 V^* \Big(\frac{yv}{z}-1 \Big) \Big] \\
	&+ k_4 (u-1) \beta_2 I_2^* \Big[S^* \Big(\frac{xw}{u}-1 \Big) + \delta_2 V^* \Big(\frac{yw}{u}-1 \Big) \Big] \\
	&+ k_5 (v-1) a_1 E_1^* \Big[\frac{z}{v}-1 \Big] + k_6 (w-1) a_2 E_2^* \Big[\frac{u}{w}-1 \Big]\\
	=& k_1 (2 \Lambda - \beta_1 S^* I_1^* - \beta_2 S^* I_2^* + \gamma V^*) + k_2 (pS^* - \delta_1 \beta_1 V^* I_1^* - \delta_2 \beta_2 V^* I_2^*) \\
	 &+ k_3 (\beta_1 S^* I_1^* + \delta_1 \beta_1 V^* I_1^*) + k_4 (\beta_2 S^* I_2^* + \delta_2 \beta_2 V^* I_2^*) + k_5 a_1 E_1^* + k_6 a_2 E_2^*\\
	 & - (k_1 \Lambda - k_1 \beta_1 S^* I_1^* - k_1 \beta_2 S^* I_2^* + k_1 \gamma V^* - k_2 p S^*)x - k_1 \Lambda \frac{1}{x}\\
	 &-(-k_1 \gamma V^* + k_2 p S^* - k_2 \delta_1 \beta_1 V^* I_1^* - k_2 \delta_2 \beta_2 V^* I_2^*)y - (k_1 \beta_1 S^* I_1^* - k_3 \beta_1 S^* I_1^*)xv\\
	 & - (-k_1 \beta_1 S^* I_1^* - k_2 \delta_1 \beta_1 V^* I_1^* + k_5 a_1 E_1^* )v - (k_1 \beta_2 S^* I_2^* - k_4 \beta_2 S^* I_2^*)xw \\
	 & - (-k_1 \beta_2 S^* I_2^* - k_2 \delta_2 \beta_2 V^* I_2^* + k_6 a_2 E_2^*)w - k_1 \gamma V^* \frac{y}{x} - k_2 p S^* \frac{x}{y} - k_3 \beta_1 S^* I_1^* \frac{xv}{z} \\
	 & - k_3 \beta_1 \delta_1 V^* I_1^*  \frac{yv}{z} - k_4 \beta_2 S^* I_2^*  \frac{xw}{u} - k_4 \beta_2 \delta_2 V^* I_2^*  \frac{yw}{u} - (k_1 \delta_1 \beta_1 V^* I_1^* - k_3 \delta_1 \beta_1 V^* I_1^*)yv \\
	 & - (k_2 \delta_2 \beta_2 V^* I_2^* - k_4 \delta_2 \beta_2 V^* I_2^*)yw - (k_3 \beta_1 S^* I_1^* + k_3 \beta_1 \delta_1 V^* I_1^* - k_5 a_1 E_1^* )z \\
	 & - (k_4 \beta_2 S^* I_2^* + k_4 \beta_2 \delta_2 V^* I_2^* - k_6 a_2 E_2^*)u - k_5 a_1 E_1^* \frac{z}{v} - k_6 a_2 E_2^* \frac{u}{w}\\
	 =: & F(x,y,z,u,v,w).
	 \end{aligned}
	\end{equation*}
	Now we choose the positive constants $k_1, k_2, k_3, k_4, k_5, k_6$ as follows:
	\begin{equation*}
	k_1 = k_2 = k_3 = k_4 =1,~ k_5 = \frac{a_1 + \mu}{a_1},~ k_6 = \frac{a_2 + \mu}{a_2}.
	\end{equation*}
	Substituting them into the function $F(x,y,z,u,v,w)$ gives
	\begin{equation*}
	\begin{aligned}
	F(x,y,z,u,v,w) =& [2\Lambda + \gamma V^* + pS^* +(a_1 + \mu)E_1^* + (a_2 + \mu)E_2^*] - \mu S^* x - \Lambda \frac{1}{x} - \mu V^* y \\
	&- \gamma V^* \frac{y}{x} - pS^* \frac{x}{y} - \beta_1 S^* I_1^* \frac{xv}{z} - \delta_1 \beta_1 V^* I_1^* \frac{yv}{z} - \beta_2 S^* I_2^* \frac{xw}{u}\\
	& - \delta_2 \beta_2 V^* I_2^* \frac{yw}{u} - (a_1 + \mu) E_1^* \frac{z}{v} -  (a_2 + \mu) E_2^* \frac{u}{w}\\
	=& \mu S^* \Big(2-x-\frac{1}{x} \Big) + \gamma V^* \Big(2- \frac{x}{y} - \frac{y}{x}\Big) + \mu V^* \Big(3- y- \frac{1}{x} - \frac{x}{y} \Big)\\
	&+ \beta_1 S^* I_1^* \Big(3- \frac{1}{x} - \frac{xv}{z} - \frac{z}{v}\Big) + \beta_2 S^* I_2^* \Big(3- \frac{1}{x} - \frac{xw}{u} - \frac{u}{w} \Big) \\
	&+ \delta_2 \beta_2 V^* I_2^* \Big(4- \frac{1}{x} - \frac{x}{y} - \frac{yw}{u} - \frac{u}{w} \Big) + \delta_1 \beta_1 V^* I_1^* \Big(4- \frac{1}{x} - \frac{x}{y} - \frac{yv}{z} - \frac{z}{v}\Big).
	\end{aligned}
	\end{equation*}
	Here, clearly $F(x,y,z,u,v,w) \leq 0,$  and the equality holds only for $x = y=1$, $z=v$, and $u=w$, i.e., $$\left\lbrace (x,y,z,u,v,w) \in \Delta : F(x,y,z,u,v,w) = 0  \right\rbrace \equiv \left\lbrace (x,y,z,u,v,w): x=y=1,z=v,u=w \right\rbrace, $$ which corresponds to the set $$\Delta' = \left\lbrace (S,V,E_1, E_2, I_1, I_2): S = S^*, V=V^*, \frac{E_1}{E_1^*}=\frac{I_1}{I_1^*}, \frac{E_2}{E_2^*}=\frac{I_2}{I_2^*} \right\rbrace \subset \Delta.$$ It is evident to see that the maximum invariant set on the set $\Delta'$ is the singleton $\left\lbrace D^* \right\rbrace $, then the endemic equilibrium $D^*$ is globally stable in $\Delta$ by LaSalle's Invariance Principle \cite{LaSalle1976}.
\end{proof}

\subsection{Bifurcations}
Different dynamical behaviors may occur in a mathematical model for the variation of the model parameters. The critical parameter value at which qualitative dynamics change occurs is called a bifurcation point. The objective of this section is to determine some local bifurcations of the system \eqref{model} with the variation of different parameters.
\subsubsection{Hopf bifurcation}
This section focuses on the local stability and Hopf bifurcation at the positive equilibrium $D^*$ of system \eqref{model}, which represents the coexistence of the both strains (native and mutant). To determine the local asymptotic stability of $D^*$, the characteristic equation of the linearized system of \eqref{model} at $D^*$ is utilized. The characteristic equation is given by 
\begin{equation}\label{chara_end}
C(\lambda) = \lambda^7 + l_1 \lambda^6 + l_2 \lambda^5 + l_3 \lambda^4 + l_4 \lambda^3 + l_5 \lambda^2 + l_6 \lambda + l_7  =0,
\end{equation} 
where $l_1, l_2, l_3, l_4, l_5, l_6,$ and $l_7$ are given in \nameref{appendix}.
Now, we define Routh-Hurwitz determinants
\begin{equation*}
\begin{aligned}
H_1 =& l_1~, \quad H_2 = \begin{array}{|cc|}
l_1 & l_3  \\ 
1 & l_2
\end{array}~, \quad
H_3 = \begin{array}{|ccc|}
l_1 & l_3  & l_5  \\ 
1 & l_2  & l_4  \\ 
0 & l_1  & l_3 
\end{array}~, \quad
H_4 = \begin{array}{|cccc|}
l_1 & l_3  & l_5 & l_7  \\ 
1 & l_2  & l_4  & l_6 \\ 
0 & l_1  & l_3 & l_5 \\
0 & 1  & l_2  & l_4 
\end{array}~, \\
H_5 =& \begin{array}{|ccccc|}
l_1 & l_3  & l_5 & l_7 & 0  \\ 
1 & l_2  & l_4  & l_6 & 0 \\ 
0 & l_1  & l_3 & l_5 & l_7 \\
0 & 1  & l_2  & l_4  & l_6 \\
0 & 0 & l_1 & l_3 & l_5
\end{array}~, \quad 
H_6 = \begin{array}{|cccccc|}
l_1 & l_3  & l_5 & l_7 & 0 & 0  \\ 
1 & l_2  & l_4  & l_6 & 0 & 0 \\ 
0 & l_1  & l_3 & l_5 & l_7 & 0 \\
0 & 1  & l_2  & l_4  & l_6 & 0 \\
0 & 0 & l_1 & l_3 & l_5 & l_7 \\
0 & 0 & 1 & l_2 & l_4 & l_6
\end{array}~, \quad H_7 = l_7 H_6.
\end{aligned}
\end{equation*}
By Routh-Hurwitz criterion, $D^*$ is locally asymptotically stable (i.e. $Re(\lambda) < 0$) if and only if $H_1 > 0, H_2 > 0, H_3 > 0, H_4 > 0, H_5 > 0, H_6 > 0,$ and $H_7 > 0;$ otherwise, $D^*$ becomes unstable.

Further, we determine the occurrence conditions of Hopf bifurcation of system \eqref{model}. To study the Hopf bifurcation, bifurcation parameter should be chosen at first. Among all parameters of system \eqref{model}, we choose the parameter $\nu_1$, which represents the mutation rate of the native strain. With the other parameter values given, we can calculate the threshold value of the bifurcation parameter $\nu_1$. In the following, we denote this threshold value of Hopf bifurcation point as $\nu_1 = \nu_1^*$.

By Liu criterion \cite{Liu1994}, we assume that there is a smooth curve of equilibrium points $(D(\nu_1), \nu_1^*)$ with $D(\nu_1) = \nu_1^*$ for system \eqref{model} and $(D^*, \nu_1^*)$ is a positive equilibrium point. If $C(\lambda, \nu_1^*) = \lambda^7 + l_1(\nu_1^*) \lambda^6 + l_2(\nu_1^*) \lambda^5 + l_3(\nu_1^*) \lambda^4 + l_4(\nu_1^*) \lambda^3 + l_5(\nu_1^*) \lambda^2 + l_6(\nu_1^*) \lambda + l_7(\nu_1^*)$ is the characteristic equation at $(D^*, \nu_1^*)$, then for a simple Hopf bifurcation, we have the following conditions:
\begin{itemize}
	\item[(i)] $l_7 (\nu_1^*)>0, H_1 (\nu_1^*) > 0, H_2 (\nu_1^*) > 0, H_3 (\nu_1^*) > 0, H_4 (\nu_1^*) > 0, H_5 (\nu_1^*) > 0,$ and $H_6 (\nu_1^*) = 0.$
	\item[(ii)] $\frac{d}{d \nu_1}(H_6 (\nu_1^*)) \neq 0,$
\end{itemize}
where $H_1 (\nu_1^*), H_2 (\nu_1^*), H_3 (\nu_1^*), H_4 (\nu_1^*), H_5 (\nu_1^*),$ and $H_6 (\nu_1^*)$ are the Hurwitz determinants at the bifurcation parameter $\nu_1^*$.

If the condition (i) holds, then the characteristic polynomial have to meet the condition for a pair of purely imaginary eigenvalues. Now, to occurrence the Hopf bifurcation, we need to derive the transversality condition (ii). For this, we let $\pm i \omega$ be a pair of purely imaginary eigenvalues. Here, differentiating the characteristic equation \eqref{chara_end} with respect to $\nu_1,$ we obtain
\begin{equation*}
\begin{aligned}
(7 \lambda^6 + 6 l_1 \lambda^5 + 5 l_2 \lambda^4 + 4 l_3 \lambda^3 &+ 3 l_4 \lambda^2 + 2 l_5 \lambda + l_6) \frac{d \lambda}{d \nu_1} + \lambda^6 \frac{d l_1}{d \nu_1} + \lambda^5 \frac{d l_2}{d \nu_1} \\
&+ \lambda^4 \frac{d l_3}{d \nu_1} + \lambda^3 \frac{d l_4}{d \nu_1} + \lambda^2 \frac{d l_5}{d \nu_1} + \lambda \frac{d l_6}{d \nu_1} + \frac{d l_7}{d \nu_1} = 0.
\end{aligned} 
\end{equation*}
Further, we obtain
\begin{equation*}
\begin{aligned}
\Bigg(\frac{d \lambda}{d \nu_1}\Bigg)^{-1} = - \frac{7 \lambda^6 + 6 l_1 \lambda^5 + 5 l_2 \lambda^4 + 4 l_3 \lambda^3 + 3 l_4 \lambda^2 + 2 l_5 \lambda + l_6}{\lambda^6 \frac{d l_1}{d \nu_1} + \lambda^5 \frac{d l_2}{d \nu_1}
	+ \lambda^4 \frac{d l_3}{d \nu_1} + \lambda^3 \frac{d l_4}{d \nu_1} + \lambda^2 \frac{d l_5}{d \nu_1} + \lambda \frac{d l_6}{d \nu_1} + \frac{d l_7}{d \nu_1}}.
\end{aligned}
\end{equation*}
Furthermore, we have
\begin{equation*}
\begin{aligned}
sign \Bigg[\frac{d(Re(\lambda))}{d \nu_1}\Bigg]_{\lambda = i\omega, H_6 = 0} &= sign \Bigg[Re \Bigg(\frac{d \lambda}{d \nu_1} \Bigg)^{-1} \Bigg]_{\lambda = i\omega, H_6 = 0} \\
& = sign [\Upsilon ], \\
\end{aligned}
\end{equation*}
where, 
\begin{equation*}
\begin{aligned}
\Upsilon & = Re \Bigg[\frac{(7\omega^6 - 5l_2 \omega^4 + 3l_4 \omega^2 - l_6) + i (-6 l_4 \omega^5 + 4l_3 \omega^2 - 2l_5 \omega)}{(-\omega^6 \frac{d l_1}{d \nu_1} + \omega^4 \frac{d l_3}{d \nu_1} -\omega^2 \frac{d l_5}{d \nu_1} + \frac{d l_7}{d \nu_1}) + i (\omega^5 \frac{d l_2}{d \nu_1} - \omega^3 \frac{d l_4}{d \nu_1} + \omega \frac{d l_6}{d \nu_1})} \Bigg] = \frac{M_1 M_3 + M_2 M_4}{M_3^2 + M_4^2 }, \\
M_1 &= 7\omega^6 - 5l_2 \omega^4 + 3l_4 \omega^2 - l_6, \quad
M_2 = -6 l_4 \omega^5 + 4l_3 \omega^2 - 2l_5 \omega, \\
M_3 &= -\omega^6 \frac{d l_1}{d \nu_1} + \omega^4 \frac{d l_3}{d \nu_1} -\omega^2 \frac{d l_5}{d \nu_1} + \frac{d l_7}{d \nu_1}, \quad
M_4 = \omega^5 \frac{d l_2}{d \nu_1} - \omega^3 \frac{d l_4}{d \nu_1} + \omega \frac{d l_6}{d \nu_1}.
\end{aligned}
\end{equation*}
If $M_1 M_3 + M_2 M_4 >0,$ then $sign \Big[\frac{d(Re(\lambda))}{d \nu_1}\Big]_{\nu_1 = \nu_1^*} >0 $ and the transversality condition (ii) holds.
Summarizing the above discussion, we obtain the subsequent theorem:
\begin{theorem}
	For the existing positive equilibrium $D^*$ of system \eqref{model}, the system \eqref{model} around $D^*$ enters into Hopf bifurcation when $\nu_1$ crosses through $\nu_1^*$.
\end{theorem} 

\subsubsection{Transcritical bifurcation}
We see that Eq. \eqref{chareq} has a zero eigenvalue when either $R_1=1$ or $R_2=1$. Thus, the system \eqref{model} may undergo a transcritical bifurcation at $D^0$ when either $R_1=1$ or $R_2=1$. In this subsection, we establish conditions on the parameters using Theorem 4.1 from Castillo-Chavez and Song \cite{Castillo2004} and center manifold theory \cite{Gukenheimer1983}. For the transcritical bifurcation, we establish the following theorem:
\begin{theorem}
	\begin{itemize}
		\item[1.] Assume $R_1<1$, the system \eqref{model} undergoes a transcritical bifurcation near $D^0$, when $R_2=1$.
		\item[2.] Assume $R_2 <1$, the system \eqref{model} undergoes a transcritical bifurcation near $D^0$, when $R_1 =1$.
	\end{itemize}
\end{theorem}
\begin{proof}
	1. We choose $\beta_2$ as a bifurcation parameter. By solving $R_2=1,$ we obtain
	\begin{equation*}
	\beta_2 = \beta_2^* = \frac{\mu  \left(a_2+\mu \right) \left(\alpha _2+d_2+\mu \right) (\gamma +\mu +p)}{a_2 \Lambda  \left(\gamma +\mu +\delta _2 p\right)}.
	\end{equation*}
	It can easily be obtained that the Jacobian $J_{(D^0, \beta_2^*)}$ evaluated at $D^0$ and $\beta_2 = \beta_2^*$ has a simple zero eigenvalue and other eigenvalues have negative sign. Hence $D^0$ is a non-hyperbolic equilibrium, when $\beta_2 = \beta_2^*.$ Now, we calculate a right eigenvector $W = (w_1, w_2, w_3, w_4, w_5, w_6)$ and a left eigenvector $V = (v_1, v_2, v_3, v_4, v_5, v_6)$ associated to the zero eigenvalues. Here
	\begin{align*}
	w_1 &= -\frac{\left(a_2+\mu \right) \left(\alpha _2+d_2+\mu \right) \left((\gamma +\mu )^2+\gamma  \delta _2 p\right)}{a_2 \mu  (\gamma +\mu +p) \left(\gamma +\mu +\delta _2 p\right)},\\
	w_2 &= -\frac{p \left(a_2+\mu \right) \left(\alpha _2+d_2+\mu \right) \left(\gamma +\mu +\delta _2 (\mu +p)\right)}{a_2 \mu  (\gamma +\mu +p) \left(\gamma +\mu +\delta _2 p\right)},\\
	w_3 & = 0, \quad w_4 = \frac{\mu + d_2 + \alpha_2}{a_2}, \quad w_5 = 0, \quad w_6 = 1, \quad w_7 = 0, \\
	v_1 & = 0, \quad v_2 = 0,\quad
	v_3 = \frac{a_1 \nu_1}{(\mu + a_1)(\mu + d_1 + \alpha_1 + \nu_1)(1-R_1)},\\ v_4 &= 1, \quad
	v_5 = \frac{\nu_1}{(\mu + d_1 + \alpha_1 + \nu_1)(1-R_1)}, \quad v_6 = 1, \quad v_7 = 0.
	\end{align*}
	Now from Theorem 4.1 of \cite{Castillo2004}, we need to calculate the bifurcation constants $a$ and $b$. For system \eqref{model}, the associated non-zero partial derivatives of $f$ (evaluated at $D^0$, $x_1 = S, x_2 = V, x_3 = E_1, x_4 = E_2, x_5 = I_1, x_6 = I_2, x_7 = R_1$) are given by  
	\begin{align*}
	a &= 2v_3 w_1 w_5 \frac{\partial^{2}f_3}{\partial S \partial I_1} + 2v_3 w_2 w_5 \frac{\partial^{2}f_3}{\partial V \partial I_1} + 2v_4 w_1 w_6 \frac{\partial^{2}f_4}{\partial S \partial I_2} 
	+ 2v_4 w_2 w_6 \frac{\partial^{2}f_4}{\partial V \partial I_2} + 2v_4 w_6 w_7 \frac{\partial^{2}f_4}{\partial I_2 \partial R_1}    \\
	& = 2 v_4 w_6(w_1 + \delta_2 w_2)\beta_2^* <0,\\
	b & = 2v_4 w_6 \frac{\partial^{2}f_4}{\partial \beta_2 \partial I_2} = 2 v_4 w_6 (S^0 + \delta_2 V^0)  > 0.
	\end{align*}
	Since the coefficient $a$ is negative and $b$ is positive, the direction of the bifurcation of system \eqref{model} at $\beta_2 = \beta_2^*$ is forward.
	\\
	2. By choosing $\beta_1$ as a bifurcation parameter. By solving $R_1=1,$ we obtain
	\begin{equation*}
	\beta_1 = \beta_1^* = \frac{\mu  \left(a_1+\mu \right) (\gamma +\mu +p) \left(\alpha _1+d_1+\mu +\nu _1\right)}{a_1 \Lambda  \left(\gamma +\mu +\delta _1 p\right)}.
	\end{equation*}
	Following simple procedure of previous case, we obtain that $D^0$ is a non-hyperbolic equilibrium, when $\beta_1 = \beta_1^*.$ Now, we calculate a right eigenvector $W_1 = (w_{11}, w_{22}, w_{33}, w_{44}, w_{55}, w_{66})$ and a left eigenvector $V_1 = (v_{11}, v_{22}, v_{33}, v_{44}, v_{55}, v_{66})$ associated to the zero eigenvalues. Here
	\begin{align*}
	w_{11} &= -\frac{\Lambda \beta_2 ((\gamma + \mu )^2+ p\gamma \delta_2)\nu_1}{\alpha_1 \mu (p+\gamma + \mu)^2 (\mu + d_2 + \alpha_2)(\mu + a_2)(1-R_2)}
	-\frac{\left(a_1+\mu \right) \left(\alpha _1+d_1+\mu +\nu _1\right) \left((\gamma +\mu )^2+\gamma  \delta _1 p\right)}{a_1 \alpha_1  (\gamma +\mu +p) \left(\gamma +\mu +\delta _1 p\right)},\\
	w_{22} &= -\frac{p(\gamma + \mu + (p+\mu)\delta_2)}{(p+\gamma + \mu)\alpha_1}\Big(\frac{\Lambda \beta_2 \nu_1}{\mu (p+\gamma + \mu)(\mu + d_2 + \alpha_2)(1-R_2)} + \frac{(\mu + a_1)(\mu + d_1 + \alpha_1 + \nu_1)}{a_1 (\gamma + \mu + p \delta_2)} \Big),\\
	w_{33} & = \frac{\mu  \left(\alpha _1+d_1+\mu +\nu _1\right)}{a_1 \alpha _1}, \quad
	w_{44} = \frac{\Lambda \beta_2 (\gamma + \mu + p\delta_2)\nu_1}{\alpha_1 (p+\gamma + \mu)(\mu + d_2 + \alpha_2)(\mu + a_2)(1-R_2)}, \quad w_{55} = \frac{\mu}{\alpha_1}, \\
	w_{66} &= \frac{\mu \nu_1}{\alpha_1 (\mu + d_2 + \alpha_2)(1-R_2)}, \quad w_{77} = 1 \quad
	v_{11}  = 0, \quad v_{22} = 0,\quad
	v_{33} = \frac{a_1}{a_1+\mu } , \quad v_{44} = 0, \\
	v_{55} &= 1, \quad v_{66} = 0, \quad v_{77} = 0.
	\end{align*}
	Similarly, as in previous case, we have
	\begin{align*}
	a &= 2v_{33} w_{11} w_{55} \frac{\partial^{2}f_3}{\partial S \partial I_1} + 2v_{33} w_{22} w_{55} \frac{\partial^{2}f_3}{\partial V \partial I_1}  \\
	& = 2v_{33} w_{11} w_{55} \beta_1^* + 2v_{33} w_{22} w_{55}\delta_1 \beta_1^*,\\
	& = 2v_{33} w_{55} \beta_1^* (w_{11} + \delta_1 w_{22})  <0,\\
	b & = 2v_{33} w_{55} \frac{\partial^{2}f_3}{\partial \beta_1 \partial I_1} = 2 v_{33} w_{55} (S^0 + \delta_1 V^0)  > 0.
	\end{align*}
	Therefore again in this case, the direction of the bifurcation of system \eqref{model} at $\beta_1 = \beta_1^*$ is forward.
\end{proof} 

\section{Implications for disease control} \label{Herd}
For vaccine-preventable diseases, not all susceptible individuals could be immunized for numerous reasons, such as they are too young to be vaccinated (vaccinating young children or infants can, sometimes, harm the children/infants), they have weakened immune system and co-morbidities (and vaccinating may make their prognosis worse), they are of higher age, or they opt out for religious or traditional reasons, cultural. However, the question is what is the smallest proportion of those individuals we can vaccinate that we need to vaccinate so that those we cannot vaccinate can be protected from emerging severe disease or dying of the disease. The idea of herd immunity in the disease dynamics is related to the indirect protection against acquiring of infectious disease, which members of the community obtain when a large percentage of the populace has become protected to the contagious disease due to natural recovery from prior infection or vaccination \cite{Anderson1992, Anderson1985, Elbasha2021}. The outcome of herd immunity is that persons who are not immune (e.g., those who have not been infected yet or cannot be vaccinated) obtain some defense against acquiring the infection. The fastest and safest and way to attain herd immunity is vaccination. It should, however, be stated that Sweden implemented the other procedure for achieving herd immunity in the COVID-19 dynamics in Sweden \cite{Friedman2020}. In other words, the Swedish public health agencies intended to achieve herd immunity by not employing the basic society transmission reduction strategies (e.g., community lockdowns, social distancing, contact tracing use of face masks in public, etc.) employed in almost every nation or community that is hard-hit with the COVID-19 pandemic, opting, instead, to let persons to attain disease and, hopefully, recover from it. In this section, a theoretical condition for achieving community-wide vaccine-induced herd immunity is obtained.
Theorem \ref{dfe_local} has significant public health implications. It reveals that if the imperfect vaccine has sufficient efficacy and coverage rate to make $R_0 < 1,$ COVID-19 will be eradicated from society. The global stability of the disease-free equilibrium (Theorem \ref{dfe_gloal}) for $R_0 < 1$ confirms that such epidemics do not hit. This means $R_0$ is an appropriate combination of parameters to measure the efficiency of a vaccination campaign.

\subsection{Herd immunity}
Not every person in a given population expects to be immunized in order to eradicate the disease. A fraction of people with immunity in the given population required to stop an epidemic is named herd immunity. Let $\rho$ represent the fraction of the vaccinated population at $D^0$ (the disease-free equilibrium). Then,
\begin{equation*}
\rho = \frac{p}{\mu + \gamma + p}.
\end{equation*}
In the absence of vaccination, i.e., when $p=0,$ the basic reproduction number is given by $R_{0wv}$. Hence, we can write
\begin{equation}\label{BRS}
R_0 = \max \left\lbrace R_{1wv}(1-(1-\delta_1)\rho),  R_{2wv}(1-(1-\delta_2)\rho) \right\rbrace, 
\end{equation}
with $R_1 = R_{1wv}(1-(1-\delta_1)\rho),$ and $R_2 = R_{2wv}(1-(1-\delta_2)\rho).$ It is noted that $R_1 \leq R_{1wv},~ R_2 \leq R_{2wv},$ and thus $R_0 \leq R_{0wv}.$ The equality holds only $\rho = 0$ (i.e., $p=0$) or $\delta_1 = \delta_2=1.$ This indicates that the vaccine, even not $100 \%$ effective, will certainly reduce the basic reproduction number of the disease. Since $R_0 < 1$ is a necessary and sufficient condition for the eradication of disease (Theorems \ref{dfe_local} and \ref{dfe_gloal}), hence it follows from \eqref{BRS} that 
\begin{equation}\label{rho_critical}
\rho > \max \left\lbrace \frac{1}{1-\delta_1}\Bigg(1-\frac{1}{R_{1wv}}\Bigg), \frac{1}{1-\delta_2}\Bigg(1-\frac{1}{R_{2wv}}\Bigg) \right\rbrace = \rho_{critical}
\end{equation}
is also a necessary and sufficient condition for disease elimination. Here, $\rho_{critical}$ signifies herd immunity. Although, this outcome could be  achieved in the case of continuous vaccination and that herd immunity is attained if the vaccination rate is large enough such that $\rho$, the fraction of vaccinated individuals at the disease-free equilibrium, exceeds the critical value $\rho_{critical}$. Combining Theorems \ref{dfe_local} and \ref{dfe_gloal}, we obtain the following consequence:
\begin{proposition}
	COVID-19 can be eradicated from the population if $\rho > \rho_{critical}.$
\end{proposition}

The inequality \eqref{rho_critical} can be expressed in terms of the vaccination rate $p$. This is done by noting, first of all, that $R_0$ is a decreasing function of $p$,
\begin{equation*}
\frac{d R_0}{dp} = \max \left\lbrace \frac{- \Lambda (\gamma + \mu)(1-\delta_1)a_1 \beta_1}{\mu(p+\gamma + \mu)^2 (\mu + a_1)(\mu + d_1 + \alpha_1 + \nu_1)}, \frac{- \Lambda (\gamma + \mu)(1-\delta_2)a_2 \beta_2}{\mu(p+\gamma + \mu)^2 (\mu + a_2)(\mu + d_2 + \alpha_2)} \right\rbrace < 0
\end{equation*}
 and so it is minimized by letting $p$ go to infinity. Taking the limit as $p$ approaches infinity, we see that this expression is always greater than $\max \left\lbrace  \delta_1 R_{1wv}, \delta_2 R_{2wv} \right\rbrace $. Thus, if $\max \left\lbrace  \delta_1 R_{1wv}, \delta_2 R_{2wv} \right\rbrace > 1,$ then no amount of vaccination can make $R_0$ smaller than unity. Alternatively, if $\max \left\lbrace  \delta_1 R_{1wv}, \delta_2 R_{2wv} \right\rbrace <1 $, then the condition
\begin{equation}\label{p_critical}
p > \max \left\lbrace \frac{(\gamma + \mu)(R_{1wv} - 1)}{1- \delta_1 R_{1wv}}, \frac{(\gamma + \mu)(R_{2wv} - 1)}{1- \delta_2 R_{2wv}} \right\rbrace = p_{critical}
\end{equation}
gives $R_0 <1$. Of course, this condition assumes $R_0>1$, since disease elimination follows without vaccination if $R_0 <1$ (by Theorems \ref{dfe_local} and \ref{dfe_gloal} and the fact that $R_0 \leq R_{0wv}$). It is easy to show that from \eqref{p_critical}, we obtain $R_0 <1$ if $p> p_{critical}$, and
$R_0 >1$ if $p< p_{critical}$. Thus, we have established the following result:
\begin{proposition}
	If $\max \left\lbrace  \delta_1 R_{1wv}, \delta_2 R_{2wv} \right\rbrace <1 $ and $p> p_{critical}$, then COVID-19 will be eliminated from
	the community. If $\max \left\lbrace  \delta_1 R_{1wv}, \delta_2 R_{2wv} \right\rbrace > 1 $, then no amount of vaccination would be able to prevent the COVID-19 outbreak in the community.
\end{proposition}

\section{Numerical illustration}\label{NumSim}
In this section, we investigate the dynamics of system \eqref{model} numerically for different sets of parameter values. Such investigations aim to determine the effect of varying the values of the different parameters and support the obtained theoretical results. It is observed that the hypothetical values of parameters given in Table \ref{num_val} are biologically feasible. However, to verify the bifurcations and different dynamical behavior of system \eqref{model}, some parameters are varied differently from Table \ref{num_val}.
\begin{table}[H]
	\caption{Numerical values of parameters.}\label{num_val}
\begin{tabular}{p{2cm}p{4cm}p{6cm}}
	\hline
	Parameters	& Value with unit & Reference \\
	\hline
	$\Lambda$	& 10000 people per week & Assumed\\
	
	$\beta_1$	& 0.000001 & Assumed \\
	
	$\beta_2$	& 0.0000003 & Assumed \\
	
	$p$	& 0.02 week$^{-1}$ & Assumed \\
	
	$\mu$	& 0.0003 week$^{-1}$ & $1/(65 \times 48)$ \\
	
	$1- \delta_1$	& 0.75 (dimensionless)  & Assumed that vaccine efficacy is 75\% to native strain \\
	
	$1- \delta_2$	& 0.40 (dimensionless) & Assumed that vaccine efficacy is 40\% to mutant strain \\
	
	$\gamma$	& 1/32 week$^{-1}$ & Assumed that loss of immunity of vaccinated individuals after 32 weeks  \\
	
	$a_1$	& 1 week$^{-1}$ &  Incubation (1 Week) \cite{Lauer2020} \\
	
	$a_2$ & 1 week$^{-1}$ & Incubation (1 Week) \cite{Lauer2020}  \\
	
	$\alpha_1$	& 1/2 week$^{-1}$ & Recovery (2 Weeks) \cite{Iboi2020} \\
	
	$\alpha_2$	& 1/2 week$^{-1}$ & Recovery (2 Weeks) \cite{Iboi2020} \\
	
	$d_1$	& 0.0006 week$^{-1}$ & Assumed \\
	
	$d_2$	& 0.0006 week$^{-1}$ & Assumed \\
	
	$1 - \delta_3$	& 0.90 (dimensionless) &  Assumed that recovered individuals have 90\% immunity against mutant strain \\
	
	$\nu_1$    & 0.3 week$^{-1}$ & Assumed \\
	
	\hline
\end{tabular}
\end{table}

The dynamics of system \eqref{model} are simulated using MATLAB 2018a. Figure \ref{stab_dfe}(a) shows that for $R_0 = \max \left\lbrace R_1, R_2 \right\rbrace = 0.9735 < 1,$ DFE $(D^0)$.is stable. Figure \ref{stab_dfe}(b) represents that the mutant dominant equilibrium ($D^2$) is stable for $R_1 <1$ and $R_2 >1$. Further, Figure \ref{diff_nu1} shows that mutation rate ($\nu_1$) can change the dynamics of positive equilibrium of system \eqref{model}. Figure \ref{diff_nu1}(a) represents that the positive equilibrium is stable. Figure \ref{diff_nu1}(b) and (c) ensure that system \eqref{model} losses its stability around the positive equilibrium for high mutation rate and periodic solution occurs. Furthermore, Figure \ref{diff_nu1}(d) shows that the positive equilibrium regains its stability for further higher mutation rate. This type of dynamics is called endemic bubble \cite{Liu2015}.

\begin{figure}[H]
	(a)\includegraphics[scale=0.5]{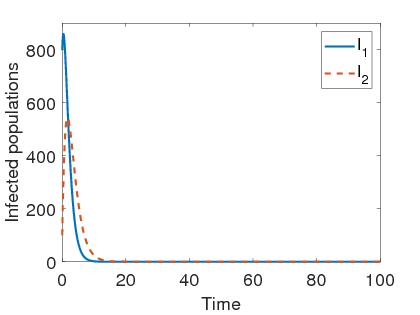}
	(b)\includegraphics[scale=0.5]{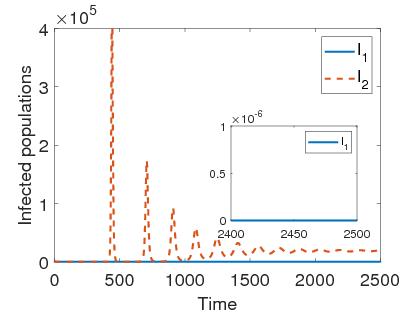}
	\caption{(a) The solution $I_1$ and $I_2$ of system \eqref{model} for the parametric values $\beta_1 = 0.000000033, \beta_2 = 0.000000017,$ and others from Table \ref{num_val}. (b) The solution $I_1$ and $I_2$ of system \eqref{model} for $\beta_1 = 0.0000000033, \beta_2 = 0.0000003,$ and other parametric values from Table \ref{num_val}.}\label{stab_dfe}
\end{figure}

\begin{figure}[H]
	(a)\includegraphics[scale=0.5]{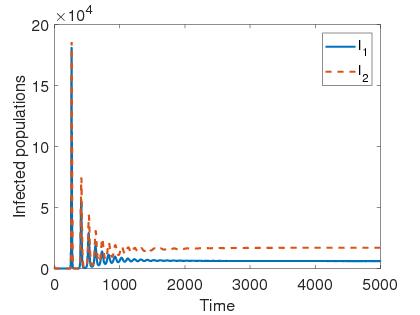}
	(b)\includegraphics[scale=0.5]{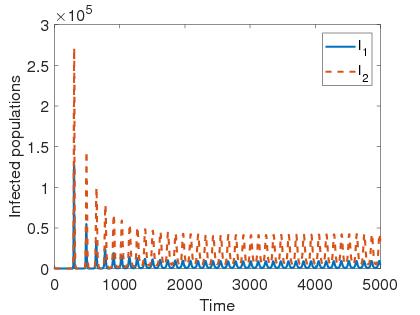}
	(c)\includegraphics[scale=0.5]{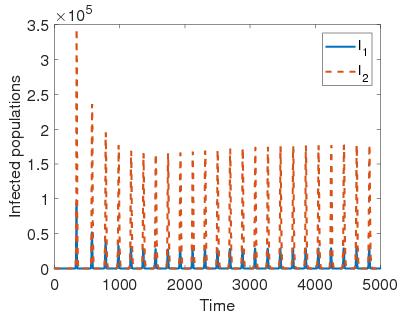}
	(d)\includegraphics[scale=0.5]{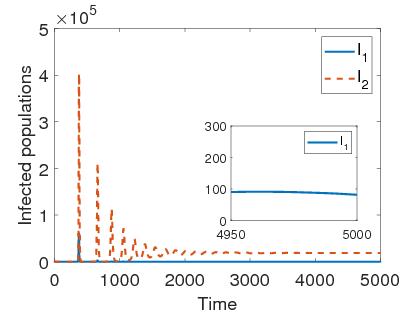}
	\caption{(a) The solution $I_1$ and $I_2$ of system \eqref{model} for the parametric values in Table \ref{num_val} and $\nu_1 = 0.3$. (b) The solution $I_1$ and $I_2$ of system \eqref{model} for $\nu_1 = 0.5$. (c) The solution $I_1$ and $I_2$ of system \eqref{model} for $\nu_1 = 0.7$. (d) The solution $I_1$ and $I_2$ of system \eqref{model} for $\nu_1 = 0.90$.} \label{diff_nu1}
\end{figure}

Furthermore, Figure \ref{bubble} shows the positive equilibrium is locally asymptotically stable for small and larger enough values of $\nu_1$, but unstable for intermediate values of $\nu_1$. In this way, we get the bifurcation diagram in Figure \ref{bubble}, which we call an endemic bubble. As shown in the bifurcation diagram (Figure \ref{bubble}), for a lower range of values of $\nu_1$, both strains persist in the environment and system \eqref{model} is asymptotically stable. For a range of $\nu_1$, the periodic oscillations (limit cycle) will appear, but for higher values of $\nu_1,$ system \eqref{model} regains its stability. Figure \ref{hopf} illustrates the Hopf bifurcation diagram with respect to the parameter $\delta_3$. It shows that the system \eqref{model} is locally asymptotically stable for lower values of $\delta_3$ and periodic solutions (Hopf bifurcation) appear for higher values of $\delta_3$.

\begin{figure}[H]
	(a)\includegraphics[scale=0.5]{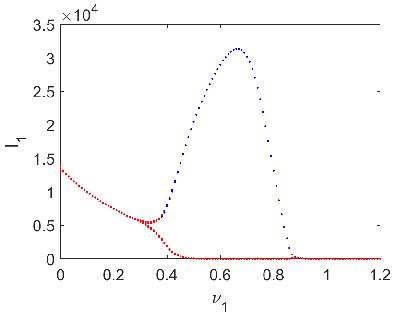}
	(b) \includegraphics[scale=0.5]{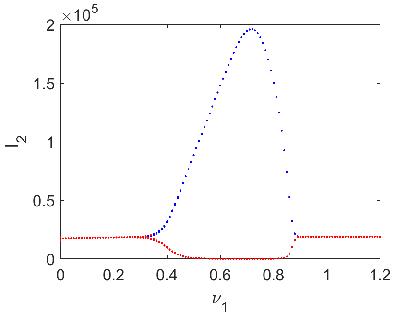}
	\caption{Bifurcation diagram (endemic bubble) with respect to mutation rate ($\nu_1$), other parametric values remain same as in  Table \ref{num_val}. The blue color shows the upper limit of the limit cycle and red color shows the lower limit of the limit  cycle.}\label{bubble}
\end{figure}

\begin{figure}[H]
	(a)\includegraphics[scale=0.5]{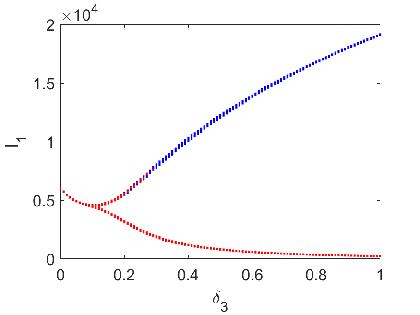}
	(b) \includegraphics[scale=0.5]{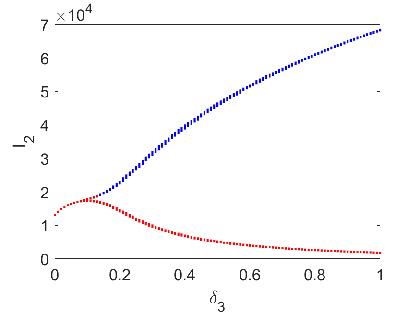}
	\caption{Hopf bifurcation diagram with respect to $\delta_3$. We keep $\nu_1 = 0.4$ and other parametric values same as in Table \ref{num_val}. The blue color shows the upper limit of the limit cycle and red color shows the lower limit of the limit cycle.}\label{hopf}
\end{figure}

\section{Case study on COVID-19 data in India} \label{cal_model}
%\subsection{Data fitting and prediction}
Here we estimate the unknown parameters of the system \eqref{model} on the cumulative cases, and cumulative deaths of COVID-19 in India from March 1, 2021 to September 27, 2021 by using the MCMC algorithm \cite{Haario2006,Haario2001,Ahmed2008}. We collect the data of cumulative cases and cumulative deaths of COVID-19 for March 1, 2021 to September 27, 2021 (31 weeks) from the WHO website \cite{dataindia}. By estimating the parameters, we estimate the mean values, standard deviation, and Geweke values of some parameters of the system \eqref{model}. The cumulative cases can be given as
\begin{equation}\label{cum_eq}
\frac{dC}{dt} = a_1 E_1 + a_2 E_2,
\end{equation}
where $C(t)$ represents the cumulative cases, and the cumulative deaths can be given as
\begin{equation}\label{cumdeath_eq}
\frac{dD}{dt} = d_1 I_1 + d_2 I_2,
\end{equation}
where $D(t)$ represents the cumulative deaths.

We use MCMC method for 20000 simulations to fit the Eq. \eqref{cum_eq} and estimate the parameters. Figure \ref{curve_fit} represents a good fitting between the cumulative reported cases of COVID-19 and the model solution, well suggesting the epidemic trend in India. Some of the parameters of the system \eqref{model} are taken either from the literature or presumed based on publicly-available COVID-19 associated information. We assume parametric values of $a_1, a_2, \mu, \delta_1, \delta_2,$ and $\delta_3$ same as given in Table \ref{num_val} and initial conditions given in Table \ref{ini_cond}. 

\begin{figure}[H]
	(a)\includegraphics[scale=0.5]{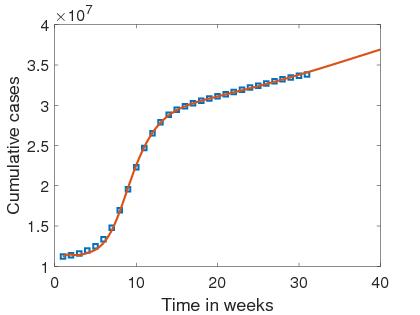}
	(b)\includegraphics[scale=0.5]{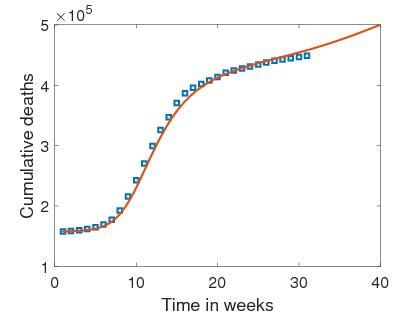}
	\caption{Fitting result of the cumulative cases and cumulative deaths of COVID-19. (a) The blue solid boxes represent the actual reported cumulative cases and the orange curve represents the model output. (b) The blue solid boxes represent the actual reported cumulative deaths and the orange curve represents the model output. For the different colors in the Figure, refer to the web version of the paper.} \label{curve_fit}
\end{figure}

By using MCMC method, we acquire the values of the parameters $\beta_1, \beta_2, \alpha_1, \alpha_2, \nu_1, d_1,$ and $d_2$ with MCMC chain of the time evolution of the cumulative cases and comparison with the confirmed cases of COVID-19 in India. We compute the mean values, standard deviation, and Geweke values of these parameters (refer in Table \ref{mcmc_param}). 

\begin{table}[H]
	\caption{Estimated values of parameters by MCMC method.}\label{mcmc_param}
	\begin{tabular}{p{2cm}p{3cm}p{2cm}p{3cm}p{2cm}}
		\hline
		Parameters	& Mean value & Standard deviation & Geweke value &  Reference \\
		\hline
		$\Lambda$	& 320000 & - & - & Assumed \\
		
		$\beta_1$	& 2.999e-07 & 5.3206e-13 & 0.99 & MCMC \\
		
		$\beta_2$	& 3.7618e-08 & 3.145e-14 & 0.95 & MCMC \\
		
		$p$	& 0.00002  & - & - & Assumed \\
		
		$\gamma$	& 0.01 & - & - & Assumed \\
		
		$\alpha_1$	& 0.45 & 5.5791e-03 & 0.99 & MCMC \\
		
		$\alpha_2$	& 0.2 & 1.6627e-03 & 0.98 & MCMC \\
		
		$\nu_1$    & 0.03 & 3.265e-04 & 0.99 & MCMC \\
		
		$d_1$    & 0.007 & 5.347e-05 & 0.99 & MCMC \\
		
		$d_2$    & 0.00003 & 3.016e-7 & 0.99 & MCMC \\
		
		\hline
	\end{tabular}
\end{table}

\begin{table}[H]
	\caption{Initial conditions for system \eqref{model} with respect to COVID-19 in India.}\label{ini_cond}
	\begin{tabular}{p{1.5cm}p{2cm}p{1.5cm}p{1.5cm}p{1.5cm}p{1.5cm}p{1.5cm}p{1.5cm}}
		\hline
		S(0)	& V(0) & $E_1(0)$ & $E_2(0)$ &  $I_1(0)$ &  $I_2(0)$ & $R_1(0)$ & $R_2(0)$\\
		\hline
		$23 \times 10^5$ & 12256337  & $10 \times 10^4$ & $60 \times 10^4$ & $10 \times 10^3$ & $80 \times 10^3$ & $60 \times 10^2$ & $20 \times 10^3$ \\
		\hline
	\end{tabular}
\end{table}

\subsection{The effects of different efficacies of vaccine and mutation rate on cumulative cases and cumulative deaths over time}
The system \eqref{model} is simulated to assess the population-level impact of the imperfect anti-COVID-19 vaccine in India. The population-level impact of the vaccine efficacy to both strains (native and mutant) on the burden of the pandemic is examined first of all. Firstly, we consider different values of the vaccine efficacy to the native strain. The system \eqref{model} is then simulated using the baseline parameter values in Table \ref{num_val}, \ref{mcmc_param}, and different values of the vaccine efficacy $(1-\delta_1)$ to the native strain. 
The results obtained, depicted in Figure \ref{imp_delta1_delta2}(a) show that, for the vaccine efficacy $(1-\delta_1) = 75 \%$ (assumed), 33,652,745 cumulative cases (the red curve) has reported by September 27, 2021. Predictions show that the cumulative cases would be recorded 40,350,000, by February 7, 2022 (31 weeks after September 27, 2021). The simulations, further, show a reduction with increasing values of the vaccine efficacy $(1-\delta_1)$ from its baseline value. In particular, if we consider the vaccine efficacy as $(1-\delta_1)= 90\%,$ then 39,030,000 cumulative cases would be recorded by February 7, 2022, representing only 3.27\% reduction. 
Figure \ref{imp_delta1_delta2}(b) represents that, for the vaccine efficacy $(1-\delta_2) = 40 \%$ (assumed) to the mutant strain, 33,652,745 cumulative cases (the yellow curve) has reported by September 27, 2021. Simulations show that the cumulative cases would be recorded 41,520,000, by February 7, 2022. These simulations show that a reduction in the cumulative cases with increasing values of the vaccine efficacy $(1-\delta_2)$. Particularly, if we increase the vaccine efficacy $(1-\delta_2) = 75 \%$, then 41,220,000 cumulative cases would be recorded by February 7, 2022. This represents 0.7\% reduction in the cumulative cases.  Furthermore, by increasing the vaccine efficacy $(1-\delta_2) = 90 \%$, the cumulative cases would be decreased 0.9\%.

The system \eqref{model} is simulated for different values of $1-\delta_3$ (immunity against mutant strain). Figure \ref{imp_delta1_delta2}(c) represents that, if recovered individuals by native strain have 90\% immunity against mutant strain ($(1-\delta_3) = 90 \%$) (assumed) to the mutant strain, 33,652,745 cumulative cases (the blue curve) has reported by September 27, 2021. Simulations show that the cumulative cases would be recorded 40,350,000, by February 7, 2022. These simulations show that an increase in the cumulative cases with decreasing value of $1-\delta_3$. Particularly, if we consider $(1-\delta_3) = 75 \%$, then 40,900,000 cumulative cases would be recorded by February 7, 2022. This represents 1.36\% increase in the cumulative cases. If we consider $(1-\delta_3) = 50 \%$, then 42,920,000 cumulative cases would be recorded by February 7, 2022, representing 6.36\% increase in the cumulative cases. Furthermore, if $(1-\delta_3) = 25 \%$, then 48,140,000 cumulative cases would be recorded by February 7, 2022, representing 19.3\% increase in the cumulative cases. 

Furthermore, simulations were also carried out to assess the impact of mutation rate on the disease dynamics of COVID-19. Figure  \ref{imp_delta1_delta2}(d) represents that, for the baseline value of the mutation rate $(\nu_1 = 0.03)$, 33,652,745 cumulative cases (the red curve) has reported by September 27, 2021. Simulations show that the cumulative cases would be 40,370,000 by February 7, 2022. These simulations show an increase in the cumulative cases with increasing values of $\nu_1$. Particularly, an increase in $\nu_1 = 0.05$, then 40,380,000 cumulative cases would be recorded by February 7, 2022. This represents a 0.02\% increase in the cumulative cases.  Furthermore, an increase in $\nu_1 = 0.09$, then 40,480,000 cumulative cases would be recorded by February 7, 2022, representing 0.27\% increase in the cumulative cases.
 
\begin{figure}[H]
	(a)\includegraphics[scale=0.5]{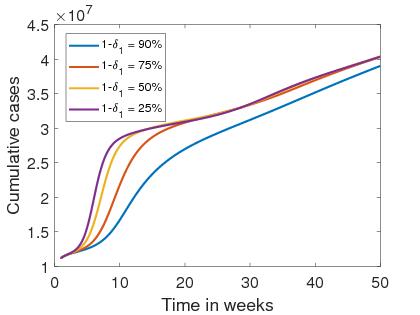}
	(b)\includegraphics[scale=0.5]{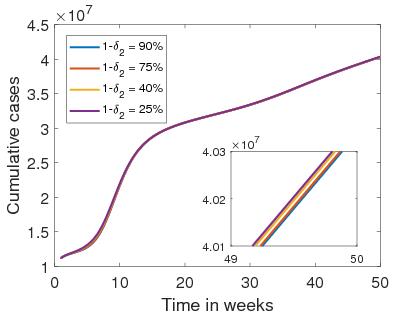}
	(c)\includegraphics[scale=0.5]{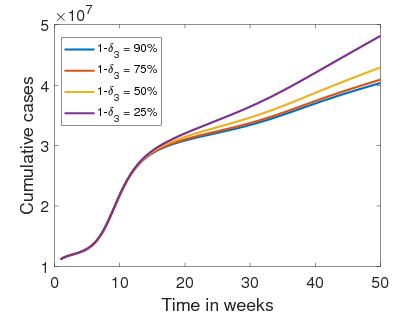}
	(d)\includegraphics[scale=0.5]{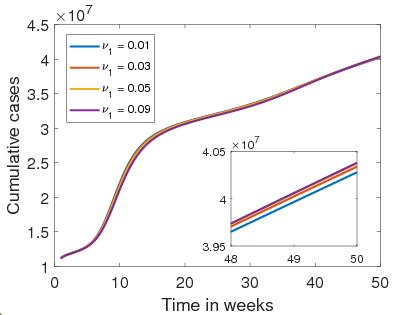}
	\caption{Assessment of the impacts of different parameters on the cumulative cases of COVID-19 in India. Simulations of the system \eqref{model} shows the cumulative cases of COVID-19 in India, as a function of time, (a) for different values of vaccine efficacy to native strain ($1-\delta_1$). (b) for different values of vaccine efficacy to mutant ($1-\delta_2$). (c) for different values of immunity against mutant strain ($1-\delta_3$). (d) for different values of mutation rate ($\nu_1$). The baseline parameter values are used from Tables \ref{num_val} and \ref{mcmc_param}.}\label{imp_delta1_delta2}
\end{figure}

To examine the dependence of the end time (February 7, 2022) of the epidemic on the vaccine efficacies to both native and mutant strains, i.e. $(1-\delta_1)$, $(1-\delta_2)$, mutation rate $(\nu_1)$, and immunity against mutant strain $(1-\delta_3)$, we sketch the contour plots of the total number of the cumulative cases of COVID-19, with respect to $(1-\delta_1)$ and $(1-\delta_2)$ in Figure \ref{count_delta1_delta2}(a); $(1-\delta_3)$ and $\nu_1$ in Figure \ref{count_delta1_delta2}(b); $(1-\delta_2)$ and $\nu_1$ in Figure \ref{count_delta1_delta2}(c); $(1-\delta_2)$ and $(1-\delta_3)$ in Figure \ref{count_delta1_delta2}(d); respectively. In order to assess the combined effect of the parameters, the other parameters remain fixed when we vary two parameters. 
The results indicate that increasing the vaccine efficacies $(1-\delta_1)$ and $(1-\delta_2)$, the cumulative cases would be reduced (Figure \ref{count_delta1_delta2}(a)). It can also be observed that the vaccine efficacy to the native strain $(1-\delta_1)$ is more influential than the vaccine efficacy to the mutant strain $(1-\delta_2)$, in controlling the total number of cases. 
In the same way, the result in Figure \ref{count_delta1_delta2}(b) shows that the cumulative number of cases would be reduced for the higher immunity against mutant strain $(1-\delta_3)$ and lower value of mutation rate $(\nu_1)$. The result also shows that cumulative cases would be higher for the lower value of $(1-\delta_3)$ and lower value of $(\nu_1)$. 
The result in Figure \ref{count_delta1_delta2}(c) shows that the cumulative number of cases would be reduced for the higher vaccine efficacy to the mutant strain $(1-\delta_2)$ and lower value of mutation rate $(\nu_1)$. The result also shows that cumulative cases would be higher for the lower value of $(1-\delta_2)$ and higher value of $(\nu_1)$. 
Figure \ref{count_delta1_delta2}(d) represents that the cumulative number of cases would be reduced only for higher immunity against mutant strain $(1-\delta_3)$. This result makes it wonder that cumulative cases would be higher even for the higher value of the vaccine efficacy to the mutant strain $(1-\delta_2)$ if there is a low value of immunity against mutant strain. This means there may exist an appropriate combination of these parameters to ensure fewer confirmed cases. Thus, it would be interesting to consider an optimal strategy for supplying vaccines to minimize the cumulative number of cases.  

\begin{figure}[H]
	(a)\includegraphics[scale=0.45]{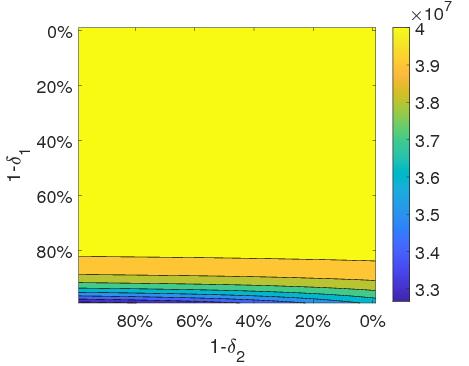}
	(b)\includegraphics[scale=0.45]{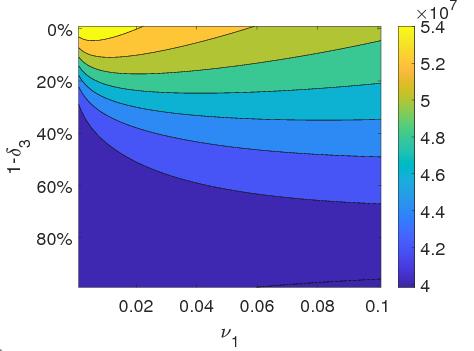}
	(c)\includegraphics[scale=0.45]{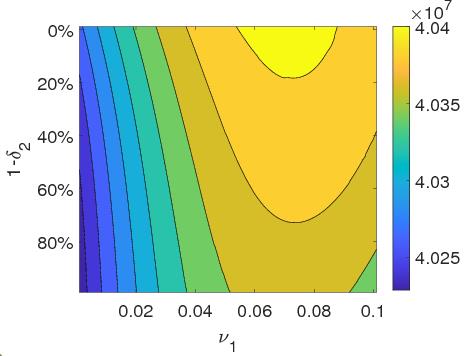}
	(d)\includegraphics[scale=0.45]{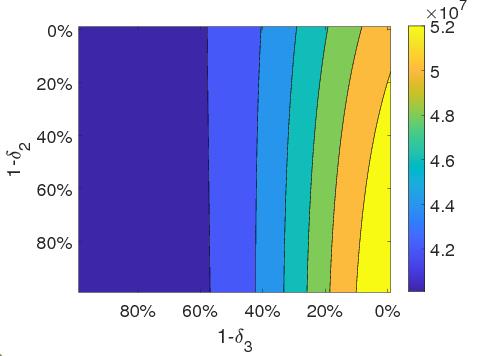}
	\caption{Contour plots of the cumulative cases, as a function of different parameters: (a) as a function of vaccine efficacy ($1-\delta_1$) to native strain and vaccine efficacy ($1-\delta_2$) to mutant strain; (b) as a function of immunity against mutant strain ($1-\delta_3$) and mutation rate ($\nu_1$); (c) as a function of vaccine efficacy ($1-\delta_2$) to mutant strain and mutation rate ($\nu_1$); (d) as a function of vaccine efficacy ($1-\delta_2$) to mutant strain and immunity against mutant strain ($1-\delta_3$). The baseline parameter values are used from Tables \ref{num_val} and \ref{mcmc_param}.}\label{count_delta1_delta2}
\end{figure}

Furthermore, the impact of vaccine efficacies on cumulative deaths is examined.
The result in Figure \ref{impdeaths_delta1_delta2}(a) shows that for the vaccine efficacy $(1-\delta_1) = 75 \%$ (assumed), 446918 cumulative deaths (red curve) has reported by September 27, 2021. Predictions show that the cumulative deaths would be recorded at 534431 by February 7, 2022 (31 weeks after September 27, 2021). The simulations further show a reduction with increasing values of the vaccine efficacy $(1-\delta_1)$ from its baseline value. In particular, if we consider the vaccine efficacy as $(1-\delta_1)= 90\%,$ then 506562 cumulative deaths would be recorded by February 7, 2022, representing only a 5.2\% reduction. 
Figure \ref{impdeaths_delta1_delta2}(b) represents that, for the vaccine efficacy $(1-\delta_2) = 40 \%$ (assumed) to the mutant strain, 446918 cumulative deaths (yellow curve) has reported by September 27, 2021. Simulations show that the cumulative deaths would be recorded at 543330 by February 7, 2022. Further simulations show a reduction in the cumulative deaths with increasing values of the vaccine efficacy $(1-\delta_2)$. Particularly, if we increase the vaccine efficacy $(1-\delta_2) = 75 \%$, then 537342 cumulative deaths would be recorded by February 7, 2022, representing a 1.1\% reduction in the cumulative deaths.  Furthermore, by increasing the vaccine efficacy $(1-\delta_2) = 90 \%$, 529420 cumulative deaths would be recorded, representing 2.5\% reduction.

\begin{figure}[H]
	(a)\includegraphics[scale=0.5]{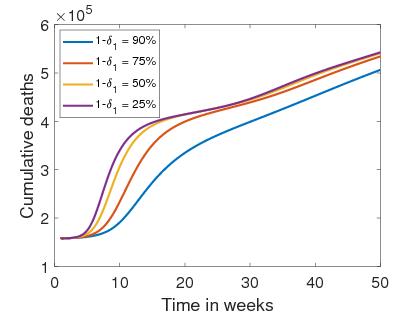}
	(b)\includegraphics[scale=0.5]{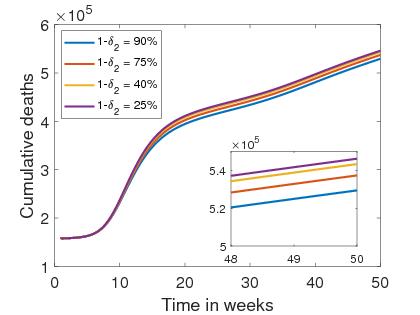}
	\caption{Assessment of the impacts of different parameters on the cumulative deaths of COVID-19 in India. Simulations of the system \eqref{model} show the cumulative deaths of COVID-19 in India, as a function of time, (a) for different values of vaccine efficacy to native strain ($1-\delta_1$); (b) for different values of vaccine efficacy to mutant ($1-\delta_2$). The baseline parameter values are used from Tables \ref{num_val} and \ref{mcmc_param}. }\label{impdeaths_delta1_delta2}
\end{figure}

\subsection{Impact of mutation rate on the dynamics of strains and infected population over time}

The impact of mutation rate ($\nu_1$) has been analyzed on the dynamics of strains and total infected population for COVID-19. We can explicitly explore how the numbers of infected individuals $I_1$ and $I_2$ depend on the mutation parameter $\nu_1$. As shown in the bifurcation diagram (Figure \ref{bif_covid}), for the range of mutation rate, $0< \nu_1 \leq 1.7$, both strains persist in the environment, and disease is asymptotically stable. For $1.7 < \nu_1 \leq 3.3$, the periodic oscillations (limit cycle) will appear, i.e. the disease outbreak will occur repeatedly. However, for a higher value of $\nu_1,$ i.e. for $\nu_1 > 3.3,$ the periodic solutions disappear and the disease again becomes stable. This dynamical phenomenon has been illustrated in Figure \ref{osci_covid}, for the values $\nu_1 = 1, 2.5$ and $3.53$. Furthermore, we have plotted the total infected population $(E_1 + E_2 + I_1 + I_2)$ over time for different values of $\nu_1$ in Figure \ref{totalinf_nu1}. This result gives a wondering dynamics over a long time. We can easily observe that if the mutation rate $(\nu_1)$ increases from 0.01 to 0.09, then the infected population persists at a lower level (blue and red curves). In addition, if the value of $\nu_1$ is increased to 2.5, then the total infected population oscillates over time (orange curve). If the value of $\nu_1$ is further increased to 3.5, then the infected population persists at a much lower level (purple curve) over long time.

\begin{figure}[H]
    (a)\includegraphics[scale=0.5]{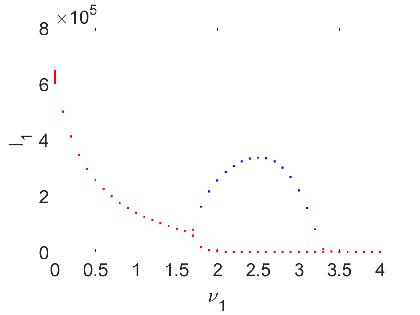}
	(b)\includegraphics[scale=0.5]{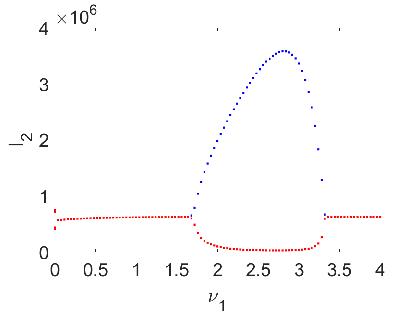} 
	\caption{Bifurcation diagram (endemic bubble) with respect to mutation rate ($\nu_1$) for the fitted values of parameters given in Table \ref{mcmc_param}. The blue color shows the upper limit of the limit cycle and red color shows the lower limit of the cycle.}\label{bif_covid}
\end{figure}

\begin{figure}[H]
	(a)\includegraphics[scale=0.5]{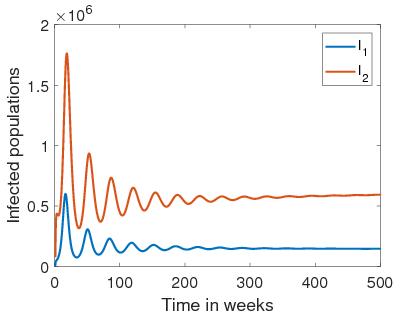}
	(b)\includegraphics[scale=0.5]{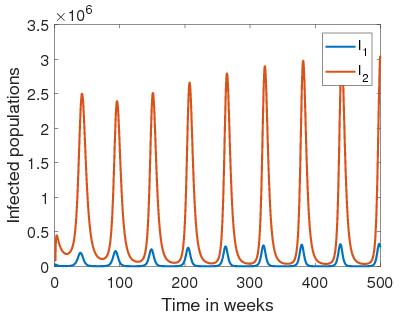}
	(c)\includegraphics[scale=0.5]{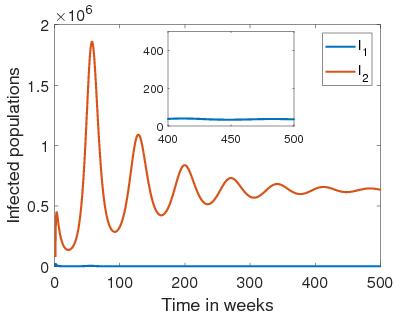}
	\caption{Impact of mutation rate on the infected populations $I_1$ and $I_2$. (a) The asymptotically stable solutions $I_1$ and $I_2$ of system \eqref{model} for the parametric values in Table \ref{mcmc_param} and $\nu_1 = 1$. (b) The periodic solutions $I_1$ and $I_2$ of system \eqref{model} for $\nu_1 = 2.5$. (c) The asymptotically stable solutions $I_1$ and $I_2$ of system \eqref{model} for $\nu_1 = 3.53$.} \label{osci_covid}
\end{figure}

\begin{figure}[H]
	\includegraphics[scale=0.5]{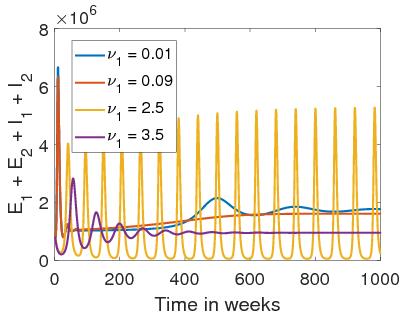}
	\caption{Variation of the total infected population with respect to different mutation rate ($\nu_1$). } \label{totalinf_nu1}
\end{figure}

\section{Sensitivity analysis}\label{Sensitivity}
Our ultimate goal in developing a mathematical model of two strains with imperfect vaccine and mutation is to find the role of different parameters to control the disease. From the viewpoints of biological significance,
$R_0$ plays a vital role in determining the severity (burden of disease),
outcome and process of the infection. This section investigates how a percentage change in key parameters in the model affects (change) the basic reproduction number. In the system \eqref{model}, the most critical parameter to reduce is the basic reproduction numbers of different strains. If the basic reproduction number is brought below one, the disease infection will be eliminated. Even if the basic reproduction number cannot be brought below one, sensitivity analysis may help to determine which parameter, if acted upon, will bring the largest reduction in the basic reproduction numbers. Sensitivity indices measure the percentage change of a key quantity, such as the basic reproduction number, in response to a percentage change of a parameter in that quantity. Sensitivity analysis is carried out on each parameter, which is utilized to recognize and check parameters responsible for impacting the basic reproductive number. The normalized sensitivity indices, also called elasticity of the quantity $Q$ with respect to the parameter $p$, is defined as follows \cite{elasticity}:
\begin{equation}
\epsilon_{p}^{Q}=\frac{\partial Q}{\partial p}\frac{p}{Q}.
\end{equation}
Elasticities can be positive or negative. A positive sign says that quantity $Q$ increases with the increase in parameter $p$, while a negative sign says that quantity $Q$ decreases with the increase of parameter $p$.
\begin{table}[H] 
	\caption{List of elasticities of $R_1$ and $R_2$.} \label{elas_R0}
	\begin{tabular}{ccc}
		\hline
		Parameters &  Values of elasticities of $R_1$ & Values of elasticities of $R_2$ \\
		\hline
		$\Lambda$	& 1 & 1\\ 
		$\beta_1 $ & 1 & -  \\
		$\beta_2$	& -  & 1 \\ 
		$a_1 $ & 0.00029 & - \\
		$a_2$	& - & 0.00029 \\ 
		$\gamma$ & 0.0014 & 0.00075 \\ 
		$\mu$	& -1 & -1 \\ 
		$p$	& -0.0014 & -0.00077 \\ 
		$\delta_1$	& 0.00048 & - \\ 
		$\delta_2$	& - & 0.00116 \\
		$d_1$	& -0.0143  &  - \\
		$d_2$ & - & -0.00015 \\
		$\alpha_1$ & -0.9234  & -  \\
		$\alpha_2$ & -  &  -0.998 \\
		$\nu_1$  & -0.0615 & - \\
		\hline
	\end{tabular}
\end{table}
We compute the elasticity indices for the basic reproduction numbers $R_1$ and $R_2$ given in Table \ref{elas_R0}. In Table \ref{elas_R0}, the parameters with positive sensitivity indices are those parameters that have a great influence on the development of the disease in the community if their values are increasing. This is because the basic reproduction number increases as their value increases; that is, the average number of secondary cases of infection increases in the community.
Also, all the parameters in which their indices are negative can curtail the infection in the community as their values increase while the others are left constant. As their values increase, the basic reproduction number decreases, which reduces the endemicity of the disease in the community.

It is evident from the values of elasticities that the reproduction number will experience the highest impact with change to recruitment rate $(\Lambda),$ transmission rates ($(\beta_1)$ and $(\beta_2)$), and natural death rate ($\mu$). In general, acting on any parameters produces a similar change in the reproduction number. For the basic reproduction number $R_0$, the elasticities concerning the parameters $\Lambda, \beta_1,$ and $\beta_2,$ are fairly similar, which means these parameters have the same impact on $R_0$.

We also perform sensitivity analysis using the methodology of
Latin Hypercube Sampling (LHS) and partial rank correlation
coefficients (PRCCs) \cite{Marino2008} to investigate the dependence
of $R_0$ on the different parameters. From Figure \ref{PRCC},
we observe that recruitment rate ($\Lambda$), transmission rates ($\beta_1, \beta_2$), natural death rate ($\mu$), and recovery rates ($\alpha_1$, $\alpha_2$)
are the most sensitive parameters for $R_0.$ To generate the LHS
matrices, we assume that all the model parameters are uniformly distributed. Then using the baseline values from Tables \ref{num_val} and \ref{mcmc_param}, a total of 1000 simulations per LHS run are carried out.

We also examine the impact of sensitivities of the parameters on the population size of cumulative cases $(C)$ and cumulative deaths $(D)$. From Figure \ref{PRCC1}, we observe that $\beta_1, a_1, \beta_2, a_2$ are the most sensitive parameters to the cumulative cases, which means that the value of these parameters increases the cumulative cases will increase. This result implies that we should control these parameters to reduce the cumulative cases. From Figure \ref{PRCC2}, we observe that $\beta_1, a_1, d_1, \alpha_1$ are the most sensitive parameters to the cumulative deaths and $\beta_1, a_1, d_1$ has positive sensitivity and $\alpha_1$ has negative sensitivity, indicating that we should control the parameters $\beta_1, a_1, d_1$ and promote $\alpha_1$ to reduce the cumulative deaths.  
\begin{figure}[H] 
	(a)\includegraphics[scale=0.5]{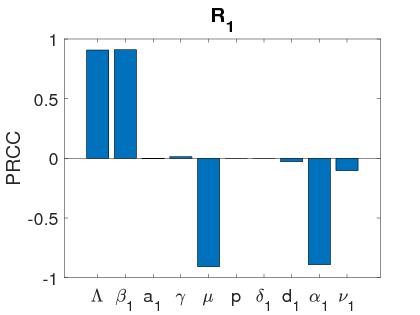}
	(b)\includegraphics[scale=0.5]{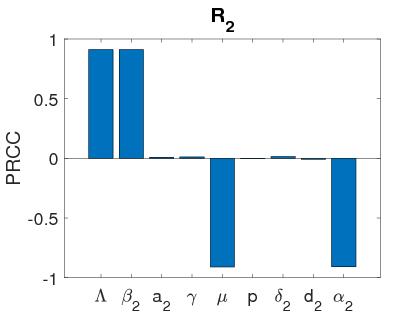}
	\caption{PRCC sensitivity on $R_0$.} \label{PRCC}
\end{figure}

\begin{figure}[H] 
	\includegraphics[scale=0.5]{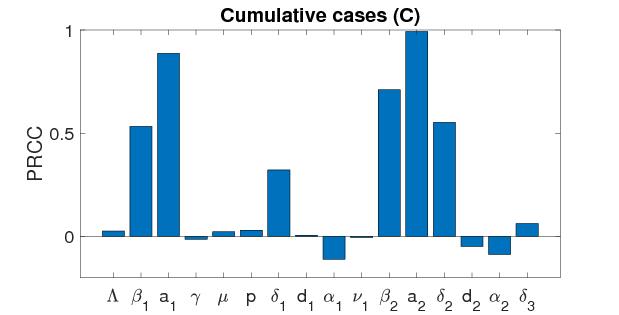}
	\caption{PRCC sensitivity on cumulative cases (C).} \label{PRCC1}
\end{figure}

\begin{figure}[H] 
	\includegraphics[scale=0.5]{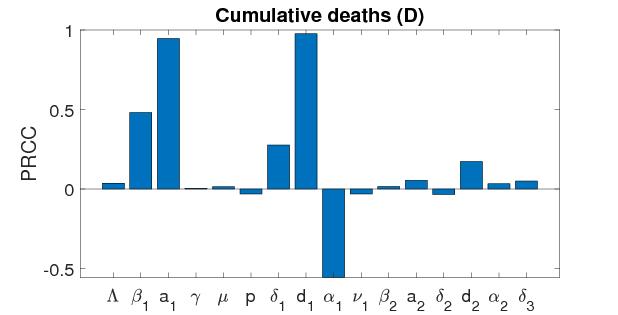}
	\caption{PRCC sensitivity on cumulative deaths (D).} \label{PRCC2}
\end{figure}

\section{Discussion}\label{Discussion}
The novel coronavirus has rapidly emerged as a disease COVID-19 and evolved as pandemics worldwide. The emergence of new variants of SARS-CoV-2 could complicate mitigation efforts. Reducing the transmission of the novel coronavirus pandemic has been the massive responsibility of the intellects of every public health agencies, Govt. officials, and millions of populations worldwide. According to the World Health Organization (WHO) \cite{dataworld}, over 26 crore populations are infected with the COVID-19 as of December 4, 2021. Now, the Govt. of different countries are trying to give safeguards to the populations via vaccination. Since every vaccine is imperfect to the disease and the virus mutates over time, mathematical models may help understand the dynamics of transmission and control of the novel coronavirus. 

Taking care of the pandemic scenario, we proposed and analyzed an SEIR type multi-strain mathematical model with imperfect vaccine and mutation \eqref{model}. We also assumed that recovered individuals of native strain could become infected by the mutant strain, but they have some immunity against the mutant strain, and it has been modeled by multiplying a reduction coefficient to the transmission rate. The model we developed takes the form of a deterministic system of nonlinear differential equations. The model is rigorously analyzed to gain insight into its dynamical features. Theoretical and numerical analysis of the proposed system \eqref{model} have been carried out using stability theory. Positivity and boundedness of the solutions of system \eqref{model} have been studied, and the system \eqref{model} is well-posed. We found that the DFE $(D^0)$ is globally stable when $R_0 < 1$ in the presence of an imperfect vaccine, i.e., the native and mutant strains will be eliminated in the community whenever $R_0 <1.$ In other words, the imperfect vaccine against COVID-19 can lead to the eradication of the pandemic if it can bring (or maintain) $R_0$ to a value less than unity. Further, we analyzed the existence and stability of the mutant dominant equilibrium ($D^2$) by constructing a suitable Lyapunov function. Global stability of coexistence (positive) equilibrium of both native and mutant strains has been investigated for the case $\nu_1=0$ and $\delta_3=0$, by constructing a suitable Lyapunov function. This result implies that if the mutation rate is zero and recovered individuals from native strain have 100\% immunity against mutant strain, the disease persists in the environment when $R_0>1$. Moreover, we investigated that system \eqref{model} undergoes a Hopf bifurcation. The transcritical bifurcation was also investigated for system \eqref{model} by using center manifold theory. By numerical simulation, we discovered that system \eqref{model} exhibits an interesting dynamics called an endemic bubble, which means the system losses its stability when the mutation rate ($\nu_1$) is larger, and oscillations occur, and further larger value of $\nu_1$, the system regains its stability (refer the Figure \ref{bubble}). The occurrence of Hopf bifurcation was also ensured by varying the parameter $\delta_3$ (Figure \ref{hopf}), which implies that the disease will appear repeatedly for the higher value of $\delta_3$. We also computed an expression for herd immunity and a threshold value for the vaccination rate that suggests disease control implications. The mentioned necessary herd immunity (induced by vaccination) percentage may not be realistically achievable. Indeed, Dr. Anthony Fauci, a member of the US Presidential Task Force on COVID-19, said on June 29, 2020, that an imaginary anti-COVID-19 vaccine might not attain the requisite high immunity in the US if many people deny getting it \cite{CableNews2020}. One way to get around this requirement for high vaccine coverage to attain high immunity is to merge the vaccination program with other anti-COVID-19 intervention strategies, such as social distancing, the use of face masks in public, etc.

Furthermore, we parameterized the proposed mathematical model using the data of the COVID-19 pandemic in India, for assessing the potential community-wide impact of an imperfect vaccine against COVID-19 and mutation of the virus. Using the data of cumulative confirmed cases and cumulative deaths of COVID-19 in India from March 1, 2021 to September 27, 2021. By employing the MCMC method to fit the model \eqref{model} with the cumulative cases and cumulative deaths, the mean values, the standard deviation, and Geweke values of the unknown parameters were estimated.

We carried out numerical simulations to measure the population-level impact of the vaccine’s efficacies. The results exhibited that the COVID-19 burden (as measured in terms of cumulative cases and cumulative deaths) decreases with increasing vaccine efficacies, as expected. We investigated how the cumulative number of cases varies with different values of vaccine efficacies to different strains, mutation rate, and immunity against the mutant strain in India. Our analysis showed what would happen if the vaccine efficacy to native strain was increased from its baseline value, as shown in Figure \ref{imp_delta1_delta2}(a).  The results showed that if the vaccine efficacy $(1-\delta_1) = 90\%$ was increased by its baseline value $(1-\delta_1) = 75\%$, then the cumulative cases would be decreased by 3.27\% by February 7, 2022. Figure \ref{imp_delta1_delta2}(b) demonstrated that if the vaccine efficacy to mutant strain increased as  $(1-\delta_2) = 75\%$ from its baseline value $(1-\delta_2) = 40\%$, then cumulative cases would be decreased by 0.7\%. We also assessed the impact of immunity against mutant strain $(1-\delta_3)$. We analyzed that if the value of $(1-\delta_3)$ is considered as  25\%, which is decreased from its baseline value $(1-\delta_3) = 90\%$, then the cumulative cases would be increased by 19.3\% by February 7, 2022, shown in \ref{imp_delta1_delta2}(c). We also analyzed the impact of mutation rate on cumulative cases, which represents the increase in cumulative cases with increasing mutation rate, showing in \ref{imp_delta1_delta2}(d). Our analysis also showed the combined impact of vaccine efficacies to both strains, which were increased to 20\%, 40\%, 60\%, 80\%, respectively, as shown in Figure \ref{count_delta1_delta2}(a). The result demonstrated that increasing the vaccine efficacy to the native strain is more influential in reducing the cumulative number of cases than the efficacy to mutant strain. We have shown in the contour plot (Figure \ref{count_delta1_delta2}(b)) that there would be higher cumulative cases for the lower immunity against mutant strain even for the lower value of the mutation rate. There will be fewer cumulative cases if immunity against mutant strain is high. The contour plot (Figure \ref{count_delta1_delta2}(c)) exhibited that cumulative cases would be high for the higher mutation rate and lower vaccine efficacy to the mutant strain. This implies that there is a need to increase in the vaccine efficacy $(1-\delta_2)$ to decrease the cumulative cases of COVID-19 in the presence of the mutation. The cumulative number of deaths has also been varied with different values of vaccine’s efficacies to different strains in Figure \ref{impdeaths_delta1_delta2}. The results in Figure \ref{impdeaths_delta1_delta2}(a) showed that if the vaccine efficacy $(1-\delta_1) = 90\%$ was increased by its baseline value $(1-\delta_1) = 75\%$, then the cumulative deaths would be decreased by 5.2\% by February 7, 2022. Figure \ref{impdeaths_delta1_delta2}(b) demonstrated that if the vaccine efficacy to mutant strain increased as  $(1-\delta_2) = 75\%$ from its baseline value $(1-\delta_2) = 40\%$, then cumulative deaths would be decreased by 1.1\%.

The endemic bubble phenomenon has been verified for COVID-19 by plotting the bifurcation diagram and periodic solutions for the fitted values from Table \ref{mcmc_param} in Figures \ref{bif_covid} and \ref{osci_covid}, respectively. These results imply that for a range of values of mutation rate ($\nu_1$), the COVID-19 outbreak will occur repeatedly. However, for the lower and higher enough values of $\nu_1$, the disease remains stable in the environment. These dynamics represent a worse scenario for the mutation of the virus in the sense of the occurrence of periodic solutions. Furthermore, the total infected population has also been plotted for different values of mutation rate $\nu_1$ and a wondering dynamics has been observed. This result shows that for a high rate of mutation, the total infected population would be stable at a lower level over a long time. However, oscillation occurs for an intermediate range of mutation rate, implying that mutation may be worse in the sense of oscillation but better in the sense of stability at a lower infection level over a long time.

Furthermore, sensitivity analysis was performed to reveal the relative significance of the key epidemiological parameters of the system \eqref{model}, which are $\Lambda, \beta_1, \beta_2, \alpha_1, \alpha_2$, because these parameters should be given priority to effectively control the disease. The PRCCs of the reproduction number (Figure \ref{PRCC}) showed that reducing $\Lambda, \beta_1, \beta_2$, which may be realized by strong control measures, such as lockdown, using face mask, travel restriction, isolation of infected individuals, contact tracing, can significantly reduce $R_0$ and thus lower the transmission risk of COVID-19. The PRCCs of the cumulative cases and deaths suggest that the parameters $\beta_1, a_1, \beta_2,$ and $a_2$ should be controlled in order to reduce the cumulative cases, and $\beta_1, a_1,$ and $d_1$ should be controlled in order to reduce the cumulative deaths.

We have to point out that the results in contour plots revealed an interesting problem, i.e., there may exist an optimal strategy of supplementing the vaccine with different efficacies, in order to ensure that the fewer cumulative cases and cumulative deaths, and that the expenditure is economical, a topic for our future work. For preserving public health, it must be focused on how the facility of medical resources availability affects the pandemic of COVID-19, which could be included in our model as future work. Formulating and analyzing the model with other non-pharmaceutical interventions \cite{Xue2022} with vaccination will be more helpful to control the spread of multi-strain disease. For instance, Iboi et al. \cite{IBOI2020} proposed a dynamical model to assess the impact of vaccine strategy with other public health intervention strategies; in particular, the authors assessed the impact of vaccine efficacy combined with mask efficacy concerning COVID-19 in the US. Their results confirmed that the elimination of COVID-19 is more feasible if the vaccine program is combined with other interventions. We focused on the situation in India, but the model can be extended to describe the efficacies of vaccines to different strains with mutation rates and developed immunity against mutant strains in other countries.

\section*{Appendix A}\label{appendix}
\begin{align*}
l_1  =& - (a_{11}+b_{12}+c_{13}+d_{13}+e_{12}+f_{13}+g_{13}), \\
l_2  =& -a_{12} b_{11}+c_{13} d_{13}+c_{13} e_{12}-c_{14} e_{11}+d_{13} e_{12}-d_{14} f_{11} +(c_{13} +d_{13}+e_{12}) f_{13} \\
&+ g_{13} (c_{13}+d_{13}+e_{12}+f_{13}) +b_{12} (c_{13}+d_{13}+e_{12}+f_{13}+g_{13}) \\
& +a_{11} (b_{12}+c_{13}+d_{13}+e_{12}+f_{13}+g_{13}), \\
l_3 = & a_{11} (-b_{12} (c_{13}+d_{13}+e_{12}+f_{13}+g_{13})-c_{13} (d_{13}+e_{12}+f_{13}+g_{13})+c_{14} e_{11} \\
&-g_{13} (d_{13}+e_{12}+f_{13})-d_{13} e_{12}-d_{13} f_{13}+d_{14} f_{11}-e_{12} f_{13})\\
& +a_{12} b_{11} (c_{13}+d_{13}+e_{12}+f_{13}+g_{13})- g_{13} (b_{12} (c_{13}+d_{13}+e_{12}+f_{13})\\
&+c_{13} (d_{13}+e_{12}+f_{13})-c_{14} e_{11}+f_{13} (d_{13}+e_{12})+d_{13} e_{12}-d_{14} f_{11})\\
& - e_{11} (a_{13} c_{11}+b_{13} c_{12}) - c_{13} d_{13} e_{12} + c_{13} d_{14} f_{11} - c_{13} d_{13} f_{13}- c_{13} e_{12} f_{13} \\
& - f_{11} (a_{14} d_{11}+b_{14} d_{12}) + d_{14} e_{12} f_{11}-d_{13} e_{12} f_{13} + c_{14} e_{11} (d_{13}+f_{13}) - d_{15} f_{11} g_{12} \\
& - b_{12} (c_{13} (d_{13}+e_{12}+f_{13})-c_{14} e_{11}+f_{13} (d_{13}+e_{12})+d_{13} e_{12}- d_{14} f_{11}),
\end{align*}
\begin{align*}
l_4 = & a_{13} e_{11} (c_{11} (b_{12}+d_{13})-b_{11} c_{12})+d_{13} (b_{12} c_{13} e_{12}-b_{12} c_{14} e_{11}+b_{13} c_{12} e_{11}) \\
& + f_{11} (a_{14} (d_{11} (b_{12}+c_{13}+e_{12})-b_{11} d_{12})-d_{14} (e_{12} (b_{12}+c_{13})+b_{12} c_{13}-c_{14} e_{11})\\
& +b_{14} d_{12} (c_{13}+e_{12})) - a_{12} (b_{11} (c_{13} (d_{13}+e_{12})-c_{14} e_{11}+d_{13} e_{12}-d_{14} f_{11})+b_{13} c_{11} e_{11}\\ &+b_{14} d_{11} f_{11}) - e_{11} f_{12} (a_{14} c_{11}+b_{14} c_{12}) + b_{12} c_{13} d_{13} f_{13} + f_{13} (a_{13} c_{11} e_{11} \\ 
&+e_{12} (d_{13} (b_{12}+c_{13})+b_{12} c_{13})-c_{14} e_{11} (b_{12}+d_{13})+b_{13} c_{12} e_{11})\\
& - a_{12} b_{11} f_{13} (c_{13}+d_{13}+e_{12}) + d_{15} f_{11} g_{12} (b_{12}+c_{13}+e_{12}) - a_{12} b_{11} g_{13} (c_{13}\\
&+d_{13}+e_{12}+f_{13}) + g_{13} (a_{13} c_{11} e_{11}+a_{14} d_{11} f_{11}+b_{12} (c_{13} (d_{13}+e_{12}+f_{13})\\
&-c_{14} e_{11}+f_{13} (d_{13}+e_{12})+d_{13} e_{12}-d_{14} f_{11})+b_{13} c_{12} e_{11}+b_{14} d_{12} f_{11}\\
&+c_{13} d_{13} e_{12}+c_{13} d_{13} f_{13}-c_{13} d_{14} f_{11}+c_{13} e_{12} f_{13}-c_{14} e_{11} (d_{13}+f_{13})\\
&+d_{13} e_{12} f_{13}-d_{14} e_{12} f_{11}) + a_{11} (b_{13} c_{12} e_{11}+b_{14} d_{12} f_{11}+c_{13} (f_{13} (d_{13}+e_{12})\\
&+d_{13} e_{12}-d_{14} f_{11})+d_{13} e_{12} f_{13}-d_{14} e_{12} f_{11}+d_{15} f_{11} g_{12}) + a_{11} (b_{12} (c_{13} (d_{13}\\
&+e_{12}+f_{13}+g_{13})-c_{14} e_{11}+d_{13} (e_{12}+f_{13}+g_{13})-d_{14} f_{11}+g_{13} (e_{12}+f_{13})\\
&+e_{12} f_{13})+g_{13} (c_{13} (d_{13}+e_{12}+f_{13})+d_{13} (e_{12}+f_{13})-d_{14} f_{11}+e_{12} f_{13})\\
&-c_{14} e_{11} (d_{13}+f_{13}+g_{13})),
\end{align*}
\begin{align*}
l_5 =& a_{12} (b_{11} (g_{13} (c_{13} (d_{13}+e_{12}+f_{13})+f_{13} (d_{13}+e_{12})+d_{13} e_{12}- d_{14} f_{11})+c_{13} d_{13} e_{12}\\
&+c_{13} d_{13} f_{13}- c_{13} d_{14} f_{11}+ c_{13} e_{12} f_{13}- c_{14} e_{11} (d_{13}+f_{13}+g_{13})+d_{13} e_{12} f_{13}\\
&- d_{14} e_{12} f_{11}+ d_{15} f_{11} g_{12})+b_{13} c_{11} e_{11} (d_{13}+f_{13}+g_{13})+b_{14} (d_{11} f_{11} (c_{13}+e_{12}+g_{13})\\
&-c_{11} e_{11} f_{12})) - a_{13} e_{11} (-b_{11} c_{12} (d_{13}+f_{13}+g_{13})+b_{12} c_{11} (d_{13}+f_{13}+g_{13})\\
&+c_{11} (g_{13} (d_{13}+f_{13})+d_{13} f_{13}- d_{14} f_{11})) - g_{13} (f_{11} (a_{14} (d_{11} (c_{13}+e_{12})- b_{11} d_{12}) \\
&+b_{14} d_{12} (c_{13}+e_{12})-c_{13} d_{14} e_{12}+ c_{14} d_{14} e_{11})- b_{12} (-a_{14} d_{11} f_{11}- c_{13} d_{13} e_{12} \\
&+ d_{14} f_{11} (c_{13}+ e_{12})+c_{14} d_{13} e_{11})-e_{11} f_{12} (a_{14} c_{11}+b_{14} c_{12})+ b_{13} c_{12} d_{13} e_{11}\\
& + f_{13} (b_{12} (c_{13} (d_{13}+e_{12})-c_{14} e_{11}+ d_{13} e_{12})+ b_{13} c_{12} e_{11}+c_{13} d_{13} e_{12}- c_{14} d_{13} e_{11}) \\
&+ a_{11} (b_{12} (c_{13} (d_{13}+e_{12}+f_{13})-c_{14} e_{11}+ f_{13} (d_{13}+e_{12})+d_{13} e_{12}-d_{14} f_{11})\\
&+b_{13} c_{12} e_{11}+b_{14} d_{12} f_{11}+c_{13} d_{13} e_{12}+c_{13} d_{13} f_{13}-c_{13} d_{14} f_{11}+c_{13} e_{12} f_{13}\\
&-c_{14} e_{11} (d_{13}+f_{13})+d_{13} e_{12} f_{13}-d_{14} e_{12} f_{11})) + a_{14} e_{11} f_{12} (c_{11} (b_{12}+d_{13})- b_{11} c_{12})\\
&+a_{14} f_{11} (b_{11} d_{12} (c_{13}+e_{12})-b_{12} d_{11} (c_{13}+e_{12})-c_{13} d_{11} e_{12}+c_{14} d_{11} e_{11})\\
&-d_{15} f_{11} g_{12} (b_{12} (c_{13}+e_{12})+c_{13} e_{12}-c_{14} e_{11})+b_{14} c_{12} d_{13} e_{11} f_{12} + f_{11} (d_{14} (b_{12} c_{13} e_{12}\\
&-b_{12} c_{14} e_{11}+b_{13} c_{12} e_{11})+b_{14} d_{12} (c_{14} e_{11}-c_{13} e_{12})) -d_{13} f_{13} (b_{12} c_{13} e_{12}- b_{12} c_{14} e_{11} \\
&+b_{13} c_{12} e_{11}) + a_{11} (-f_{13} (b_{12} (c_{13} (d_{13}+e_{12})-c_{14} e_{11}+d_{13} e_{12})+b_{13} c_{12} e_{11}+ c_{13} d_{13} e_{12}\\
&-c_{14} d_{13} e_{11})-b_{12} c_{13} d_{13} e_{12}+ b_{12} c_{13} d_{14} f_{11}-d_{15} f_{11} g_{12} (b_{12}+c_{13}+e_{12})+b_{12} c_{14} d_{13} e_{11} \\
&+b_{12} d_{14} e_{12} f_{11}- b_{13} c_{12} d_{13} e_{11}+b_{14} c_{12} e_{11} f_{12}-b_{14} c_{13} d_{12} f_{11}-b_{14} d_{12} e_{12} f_{11}\\
&+c_{13} d_{14} e_{12} f_{11}- c_{14} d_{14} e_{11} f_{11}),
\end{align*}
\begin{align*}
l_6 =& a_{13} e_{11} (-g_{13} (b_{11} c_{12} (d_{13}+f_{13})-c_{11} d_{13} f_{13}+c_{11} d_{14} f_{11})-b_{11} c_{12} d_{13} f_{13}+b_{11} c_{12} d_{14} f_{11}\\
&+b_{12} c_{11} (g_{13} (d_{13}+f_{13})+d_{13} f_{13}-d_{14} f_{11})+b_{14} f_{11} (c_{11} d_{12}-c_{12} d_{11})+c_{11} d_{15} f_{11} g_{12}) \\
& +  a_{14} (b_{11} c_{12} d_{13} e_{11} f_{12}-b_{11} c_{13} d_{12} e_{12} f_{11}+b_{11} c_{14} d_{12} e_{11} f_{11}-b_{12} (c_{11} d_{13} e_{11} f_{12} \\
&-c_{13} d_{11} e_{12} f_{11}+c_{14} d_{11} e_{11} f_{11})+b_{13} e_{11} f_{11} (c_{12} d_{11}-c_{11} d_{12})-c_{11} d_{15} e_{11} f_{11} g_{11} \\
&+ g_{13} (f_{11} (-b_{11} d_{12} (c_{13}+e_{12})+b_{12} d_{11} (c_{13}+e_{12})+c_{13} d_{11} e_{12}-c_{14} d_{11} e_{11}) \\
&-e_{11} f_{12} (c_{11} (b_{12}+d_{13})-b_{11} c_{12}))) - g_{13} (a_{11} (b_{12} d_{13} (c_{14} e_{11}-c_{13} e_{12}) \\
&+b_{12} d_{14} f_{11} (c_{13}+e_{12})-b_{13} c_{12} d_{13} e_{11}+b_{14} c_{12} e_{11} f_{12}-f_{11} (b_{14} d_{12} (c_{13}+e_{12}) \\
&-c_{13} d_{14} e_{12}+c_{14} d_{14} e_{11}))+d_{14} f_{11} (b_{12} c_{13} e_{12}-b_{12} c_{14} e_{11}+b_{13} c_{12} e_{11}) \\
&+b_{14} (c_{12} d_{13} e_{11} f_{12}-c_{13} d_{12} e_{12} f_{11}+c_{14} d_{12} e_{11} f_{11}) + a_{12} (b_{11} (c_{13} (f_{13} (d_{13}+e_{12})\\
&+d_{13} e_{12}-d_{14} f_{11})-c_{14} e_{11} (d_{13}+f_{13})+d_{13} e_{12} f_{13}-d_{14} e_{12} f_{11})+b_{13} c_{11} e_{11} (d_{13}+f_{13})\\
&-b_{14} c_{11} e_{11} f_{12}+b_{14} d_{11} f_{11} (c_{13}+e_{12})) - f_{13} (a_{11} (b_{12} (e_{12} (c_{13}+d_{13})+c_{13} d_{13}-c_{14} e_{11}) \\
&+b_{13} c_{12} e_{11}+c_{13} d_{13} e_{12}-c_{14} d_{13} e_{11})+d_{13} (b_{12} c_{13} e_{12}-b_{12} c_{14} e_{11}+b_{13} c_{12} e_{11}))) \\
&+ a_{11} (-(d_{14} f_{11}-d_{13} f_{13}) (b_{12} c_{13} e_{12}-b_{12} c_{14} e_{11}+b_{13} c_{12} e_{11})+d_{15} f_{11} g_{12} (e_{12} (b_{12}+c_{13}) \\
&+b_{12} c_{13}-c_{14} e_{11})-b_{14} (c_{12} d_{13} e_{11} f_{12}-c_{13} d_{12} e_{12} f_{11}+c_{14} d_{12} e_{11} f_{11}))+a_{12} ((d_{14} f_{11} \\
&-d_{13} f_{13}) (b_{11} c_{13} e_{12}-b_{11} c_{14} e_{11}+b_{13} c_{11} e_{11})-b_{11} d_{15} f_{11} g_{12} (c_{13}+e_{12}) \\
&+b_{14} (c_{11} d_{13} e_{11} f_{12}-c_{13} d_{11} e_{12} f_{11}+c_{14} d_{11} e_{11} f_{11}))+d_{15} f_{11} (g_{12} (b_{12} c_{13} e_{12}-b_{12} c_{14} e_{11}\\
&+b_{13} c_{12} e_{11})-b_{14} c_{12} e_{11} g_{11}),
\end{align*}
\begin{align*}
l_7 = & a_{14} d_{15} e_{11} f_{11} g_{11} (b_{12} c_{11}-b_{11} c_{12}) + g_{13} (a_{11} ((d_{14} f_{11}-d_{13} f_{13}) (b_{12} c_{13} e_{12}-b_{12} c_{14} e_{11}\\
&+b_{13} c_{12} e_{11})+b_{14} (c_{12} d_{13} e_{11} f_{12}-c_{13} d_{12} e_{12} f_{11}+c_{14} d_{12} e_{11} f_{11}))-a_{12} (d_{14} f_{11} (b_{11} c_{13} e_{12}\\
&-b_{11} c_{14} e_{11}+b_{13} c_{11} e_{11})+b_{14} (c_{11} d_{13} e_{11} f_{12}-c_{13} d_{11} e_{12} f_{11}+c_{14} d_{11} e_{11} f_{11}))\\
&+a_{12} d_{13} f_{13} (b_{11} c_{13} e_{12}-b_{11} c_{14} e_{11}+b_{13} c_{11} e_{11})+a_{13} e_{11} ((b_{12} c_{11}-b_{11} c_{12}) (d_{14} f_{11} \\
&-d_{13} f_{13})+b_{14} f_{11} (c_{12} d_{11}-c_{11} d_{12})))+d_{15} f_{11} (g_{12} (e_{11} (a_{11} b_{12} c_{14}-a_{11} b_{13} c_{12}\\
&-a_{12} b_{11} c_{14}+a_{12} b_{13} c_{11}+a_{13} b_{11} c_{12}-a_{13} b_{12} c_{11})+c_{13} e_{12} (a_{12} b_{11}-a_{11} b_{12}))\\
&+b_{14} e_{11} g_{11} (a_{11} c_{12}-a_{12} c_{11}))+a_{14} g_{13} (d_{13} e_{11} f_{12} (b_{12} c_{11}-b_{11} c_{12})+f_{11} (b_{12} d_{11}\\
&-b_{11} d_{12}) (c_{14} e_{11}-c_{13} e_{12})+b_{13} e_{11} f_{11} (c_{11} d_{12}-c_{12} d_{11})).
\end{align*}
\begin{align*}
a_{11} = & -\beta_1 I_1^* - \beta_2 I_2^* - (\mu + p), \quad
a_{12} = \gamma, \quad
a_{13} = -\beta_1 S^*, \quad
a_{14} = -\beta_2 S^*, \\
b_{11} = & p, \quad
b_{12} = -\delta_1 \beta_1 I_1^* - \delta_2 \beta_2 I_2^* - (\mu + \gamma), \quad
b_{13} = -\delta_1 \beta_1 V^*,\\
b_{14} = & -\delta_2 \beta_2 V^*, \quad
c_{11} = \beta_1 I_1^*, \quad
c_{12} = \delta_1 \beta_1 I_1^*, \quad
c_{13} = -(a_1 + \mu),\\
c_{14} = & \beta_1 (S^* + \delta_1 V^*), \quad
d_{11} = \beta_2 I_2^*, \quad
d_{12} = \delta_2 \beta_2 I_2^*, \quad
d_{13} = -(a_2 + \mu),\\
d_{14} = & \beta_2 (S^* + \delta_2 V^*) + \delta_3 \beta_2 R_1^*, \quad
d_{15} = \delta_3 \beta_2 I_2^*, \quad
e_{11} = a_1, \\
e_{12} = & -(\alpha_2 + \mu + d_1 + \nu_1), \quad
f_{11} = a_2, \quad
f_{12} = \nu_1, \quad
f_{13} = -(\alpha_2 + \mu + d_2),\\
g_{11} = & \alpha_1, \quad
g_{12} = -\delta_3 \beta_2 R_1^*, \quad
g_{13} = -(\delta_3 \beta_2 I_2^* + \mu).
\end{align*}
\end{document}